\documentclass[11pt, oneside]{article}   	
\usepackage{geometry} 
\usepackage{color}               		
\usepackage{graphicx}				
\usepackage{amssymb}
\usepackage{setspace}

\usepackage[title]{appendix}   

\usepackage{authblk}

\usepackage{mathrsfs}	

\usepackage[semicolon]{natbib}

\usepackage{verbatim}
\usepackage{soul}	
\usepackage{graphicx}
\usepackage{subcaption}
\usepackage{hyperref}

\usepackage{multirow}

\usepackage{esdiff} 
\usepackage{amsmath}
\usepackage{amsthm}

\usepackage{enumitem}   

\usepackage{etoolbox}   
\usepackage{imakeidx}   
\makeindex              

\usepackage{algpseudocode}
\usepackage{algorithm} 

\usepackage{caption}    
\newcommand{\bF}{\mathbb{F}}
\newcommand{\bG}{\mathbb{G}}

\newcommand{\bQ}{\mathbb{Q}}
\newcommand{\bV}{\mathbb{V}}
\newcommand{\bE}{\mathbb{E}}
\newcommand{\bR}{\mathbb{R}}

\DeclareMathOperator*{\argmax}{argmax}

\newcommand{\hq}{\hat{Q}}

\newcommand{\mh}{Metropolis-Hastings}

\newtheorem{proposition}{Proposition}[section]

\title{A review of Monte Carlo-based versions of the EM algorithm}
\author[1]{William Ruth}
\affil[1]{Corresponding Author - Department of Statistics and Actuarial Science \\ Simon Fraser University \\ Burnaby, BC  Canada \\ wruth@sfu.ca}
\date{}

\begin{document}
\maketitle


\begin{abstract}
    The EM algorithm is a powerful tool for maximum likelihood estimation with missing data. In practice, the calculations required for the EM algorithm are often intractable. We review numerous methods to circumvent this intractability, all of which are based on Monte Carlo simulation. We focus our attention on the Monte Carlo EM (MCEM) algorithm and its various implementations. We also discuss some related methods like stochastic approximation and Monte Carlo maximum likelihood. Generating the Monte Carlo samples necessary for these methods is, in general, a hard problem. As such, we review several simulation strategies which can be used to address this challenge.

    Given the wide range of methods available for approximating the EM, it can be challenging to select which one to use. We review numerous comparisons between these methods from a wide range of sources, and offer guidance on synthesizing the findings. Finally, we give some directions for future research to fill important gaps in the existing literature on the MCEM algorithm and related methods.
\end{abstract}

\section{Introduction}

The EM algorithm \citep{Dem77} is a very influential method for the analysis of missing data. In fact, the original paper detailing this method has ranked among the most cited papers both within the statistics literature \citep{Rya05}, and among science as a whole \citep{Van14}. The EM algorithm is an iterative method which can be used when the observed data are framed as a partial observation from some unobserved `complete' dataset. Unfortunately, there are certain technical barriers which can make the EM algorithm intractable in practice. One such challenge is how to compute conditional expectations given the observed data. A popular way to address this challenge is by replacing conditional expectations with Monte Carlo averages. We refer to such a method which makes this substitution as a Monte Carlo EM (MCEM) algorithm.

The practical implementation of the MCEM algorithm is plagued by numerous obstacles, including when to terminate and what Monte Carlo sample size to use. It is also necessary to devise a way to simulate from the appropriate conditional distribution. Numerous authors have proposed solutions to the termination and sample size problems, although there are is limited guidance for choosing between these solutions.


The MCEM algorithm was originally proposed in a paper by \citet{Wei90}, which introduces the method and suggests some simple convergence diagnostics. Follow-up work by \citet{Cha95} uses a pilot study to quantify Monte Carlo uncertainty and inform the design of a follow-up analysis. Related work by \citet{Boo99} and \citet{Caf05} focus on uncertainty quantification for the MCEM algorithm as an approximation to the EM algorithm at each step.

Some alternative methods also exist which solve the same problem as the MCEM algorithm but do not fit as neatly within the EM algorithm framework. One such method is stochastic approximation \citep{Gu98I,Del99,Rob51}, which, much like MCEM, is an iterative algorithm which requires that a Monte Carlo sample be generated at each iteration, but uses a different update formula from the EM algorithm. Another alternative is Monte Carlo maximum likelihood \citep{Gey91}, which uses Monte Carlo averaging to approximate the entire likelihood function of the observed data and estimates the parameter of interest by maximizing this approximate likelihood.

We illustrate all the MCEM algorithms and related methods on a running example. Computation for these analyses is done in \texttt{Julia} \citep{Bez17}. Our code is available in a GitHub repository \citep{Rut23I}.


While the idea of the MCEM algorithm sounds promising---replace intractable conditional expectations with Monte Carlo averages---generating the required Monte Carlo samples can also be challenging. Fortunately, numerous methods exist for simulating from difficult distributions, such as importance sampling \citep{Rob04} and Markov chain Monte Carlo \citep{Gel13}. We focus our discussion of simulation strategies on these two methods, but also briefly touch on rejection sampling \citep{Rob04}, sequential Monte Carlo \citep{Del06} and quasi-Monte Carlo \citep{Caf98}.

Numerous comparisons have been made between the various methods discussed above. Some of these comparisons are made in the context of proposing a new method \citep[e.g.,][]{Gu98I,Boo99}, while others are full papers dedicated solely to comparing various methods \citep[e.g.,][]{McC97,Boo01}. We review these comparisons and give some guidance on how to synthesize the sometimes conflicting findings. 


Other authors have various aspects of the MCEM algorithm and related work. \citet{Cel95} reviews several simulation-based methods from a slightly different perspective than the one we focus on here, and gives a comparison on simulated data. See our discussion of \citet{Cel95} in Section \ref{sec:SAEM} for more on this difference. Chapters 10 and 11 of \citet{Cap05} give a textbook-level overview of the EM algorithm and some of its Monte Carlo-based extensions. \citet{Jan06II} reviews numerous aspects of the MCEM algorithm, with a particular focus on considerations required for its implementation. \citet{Nea13} reviews several contributions to the theoretical basis for the MCEM algorithm. \citet{Zho20} discusses the use of MCEM and related methods on missing data problems. In this work, we extend these existing reviews, both by incorporating some more recent developments, and by compiling and synthesizing the various numerical comparisons scattered throughout the literature. We also offer some direction for future work.

The rest of this paper is organized as follows. Section \ref{sec:EM} contains a brief overview of the EM algorithm and some of its properties which are most relevant to the MCEM algorithm. Section \ref{sec:MCEM} reviews different implementations of the MCEM algorithm. Section \ref{sec:alternatives} presents alternatives to MCEM. Section \ref{sec:simulation} details some simulation strategies. In Section \ref{sec:comparison} we discuss comparisons between the MCEM methods and their alternatives. Finally, Section \ref{sec:conc} gives our conclusions and some directions for future work.

\section{The EM Algorithm}
\label{sec:EM}

We begin by setting-up the missing data framework and setting some important terminology which will be used throughout our work. Let $Y$ be the observed data and $X$ be the missing data. Note that $X$ need not correspond to any actual real-world process, but may instead be a conceptual device which facilitates analysis of the data that were actually observed. We refer to the distribution of $Y$ as the ``observed data distribution'', and write $f$ for its density (or mass function). We refer to the joint distribution of $Y$ and $X$ as the ``complete data distribution''\index{Complete Data Distribution}, and write $f_c$ for its density. We refer to the conditional distribution of $X$ given $Y$ as the ``missing data distribution''\index{Missing Data Distribution}, and write $f_m$ for its density. Note that the missing data distribution is not the marginal distribution of the missing data, but rather its conditional distribution given the observed data. When discussing conditional expectations, we always refer to the conditional expectation of the missing data given the observed data unless stated otherwise.

We also write $\ell$, $\ell_c$ and $\ell_m$ for the log-likelihoods based on the observed, complete and missing data distributions respectively. Similarly, we use $S$, $S_c$, $S_m$ for the scores (gradients of the corresponding $\ell$'s) and $I$, $I_c$, $I_m$ for the observed information matrices (negative Hessians of the corresponding $\ell$'s). We emphasize that, in our notation, a subscript $c$ denotes ``complete'' rather than ``conditional'', and a subscript $m$ denotes ``missing'' rather than ``marginal''. We write $\theta \in \Theta \subseteq \bR^p$ for the parameter of interest, and note that $f$, $f_c$ and $f_m$ are parameterized by the same $\theta$ (although in principle they need not all depend on every component of $\theta$)

The EM algorithm is a method for analyzing incomplete data which was formalized by \citet{Dem77}. See \citet{McL08} for an excellent book-length overview of the EM algorithm. We first describe the EM algorithm, then give some of its properties. We also illustrate EM by analysing a toy problem of inferring genotype frequencies based on measured blood group phenotypes.

The EM algorithm consists of iterating two steps. First is the expectation, or ``E'', step, in which an objective function is constructed from the complete data likelihood. Second is the maximization, or ``M'', step, in which the previously computed objective function is maximized. These two steps are then alternated until some convergence criterion is met. Whatever value of $\theta$ the algorithm converges to is used as our parameter estimate. We now go into more detail on each of the two steps.

The E-step of the EM algorithm consists of computing the conditional expectation of the complete data likelihood, given the observed data. That is, the objective function at iteration $k$ is given by
\begin{align}
    Q(\theta|\theta_{k-1}) & = \bE_{\theta_{k-1}}[\ell_c(\theta; y, X) | Y=y]
\end{align}
where $\theta_{k-1}$ is the parameter estimate obtained from the previous iteration. 

The M-step of the EM algorithm consists of maximizing the objective function constructed in the previous E-step. That is, we define $\theta_k = \argmax\limits_\theta Q(\theta|\theta_{k-1})$. Typically, this optimization must be performed numerically via, e.g., gradient ascent or the Newton-Raphson algorithm. See \citet{Noc06} for details and other optimization methods. In fact, it is possible to divide the set of parameters into groups (possibly with each group containing a single parameter) and optimize over each group individually with the others held fixed. This is called the Expectation-Conditional Maximization, or ECM, algorithm \citep{Men93}, and can sometimes be used to exploit the structure of a problem to accelerate convergence.

Notationally, we can combine the E- and M-steps of the EM algorithm into a single ``update function''. We write $M(\theta_{k-1}) = \argmax\limits_\theta Q(\theta|\theta_{k-1})$. The EM algorithm can thus be viewed as the iterative application of this update function, $M$. Note that EM performs local, but not necessarily global, optimization; although, some work has been done toward modifying the EM algorithm and related methods to perform global optimization \citep{Jan06II, Jan06III}. Furthermore, it is common in the types of problems to which the EM algorithm is applied for the (observed data) likelihood surface to be multi-modal \citep{McL08}. Thus, our choice of where to start our iteration can be very important.

\subsection{Properties}

In this section, we discuss some of the main properties of the EM algorithm. The EM algorithm literature is vast, so we present only a few of the highlights which will be most important to us later.

\subsubsection{Ascent Property and Generalized EM}
\label{sec:GEM}

An important feature of the EM algorithm is its so-called ``ascent property''. This property says that an iteration of the EM algorithm never results in a decrease in the observed data likelihood. This fact is somewhat surprising upon first reading, since updates are computed without ever evaluating the observed data likelihood. 

\begin{proposition}[Ascent Property of EM]
    \label{thm2:EM_ascent}
    Let $\theta \in \Theta$, and $\theta' = M(\theta)$ be the EM update from $\theta$. Then $\ell(\theta') \geq \ell(\theta)$.
\end{proposition}

\begin{proof}
    We begin by noting that the following decomposition holds for any value of $x$:
    \begin{align}
        \ell(\theta; y) &= \ell_c(\theta; y, x) - \ell_m(\theta; y, x)
    \end{align}
    Subtracting the values of both sides at $\theta$ from their values at $\theta'$ and taking conditional expectations, we get
    \begin{align}
        \ell(\theta'; y) - \ell(\theta; y) &= Q(\theta'|\theta) - Q(\theta|\theta) + \bE_{\theta}[\ell_m(\theta; y, x) - \ell_m(\theta'; y, x)]\\
        &= Q(\theta'|\theta) - Q(\theta|\theta) + \mathrm{KL}(\theta \rightarrow \theta') \label{eq:asc_KL}
    \end{align}
    where the last term in line (\ref{eq:asc_KL}) is the Kullback-Leibler (KL) divergence from the missing data distribution with $\theta = \theta$ to the same distribution with $\theta = \theta'$ \citep{van98}. Note that KL divergences are always non-negative, so we get
    \begin{align}
    \ell(\theta'; y) - \ell(\theta; y) &\geq Q(\theta'|\theta) - Q(\theta|\theta)  \label{eq:EM_ascent}  
    \end{align}
    Finally, since $\theta'$ maximizes $Q(\cdot|\theta)$, we have $\ell(\theta'; y) - \ell(\theta; y) \geq 0$.
\end{proof}

Note that, as long as our model is identifiable, the inequality in line (\ref{eq:EM_ascent}) is strict when $\theta \neq \theta'$. Additionally, in our proof of Proposition \ref{thm2:EM_ascent}, we only required that $Q(\theta'|\theta) \geq Q(\theta|\theta)$, not that $\theta'$ maximize $Q(\cdot|\theta)$. This observation leads to the definition of the ``Generalized EM algorithm'', which replaces the M-step with setting $\theta_k$ to any point in $\Theta$ such that $Q(\theta_k|\theta_{k-1}) \geq Q(\theta_{k-1}|\theta_{k-1})$.

\subsubsection{Recovering Observed Data Likelihood Quantities}

Under regularity conditions \citep[see][]{McL08}, it is possible to compute both the score vector and the observed information matrix of the observed data likelihood using complete data quantities. These regularity conditions consist mostly of being able to interchange the order of differentiation and integration for various functions.

\begin{proposition}
    \label{thm2:EM_decomp}
    Provided that differentiation and integration can be exchanged and that all given expectations are finite, the following identities hold:
    \begin{enumerate}[label=(\roman*)]
        \item $S(\theta; y) = \bE_\theta [S_c(\theta; y, X)|Y=y]$ \label{eq:obs_score_identity}
        \item $I(\theta) = \mathcal{I}_c(\theta) - \mathcal{I}_m(\theta)$ \label{eq:obs_info_identity}\\
        where $\mathcal{I}_c(\theta) := - \bE_\theta \left[ \nabla^2 \ell_c(\theta; y, X) | Y=y \right]$ and $\mathcal{I}_m(\theta) := - \bE_\theta \left[ \nabla^2 \ell_m(\theta; y, X) | Y=y \right]$
    \end{enumerate}
\end{proposition}

\begin{proof}
    We start with expression (\ref{eq:obs_score_identity}). Let $\Omega$ be the complete data sample space. Let $\mathcal{Y}$ and $\mathcal{X}$ be the observed and missing data sample spaces respectively. For every $y \in \mathcal{Y}$, let $\mathcal{X}(y) = \{ x \in \mathcal{X}: (y,x) \in \Omega\}$. Note that $f(y; \theta) = \int_{\mathcal{X}(y)} f_c(y, x; \theta) dx$.
    \begin{align}
        \bE_\theta [S_c(\theta; y, X)|Y=y] &= \int_{\mathcal{X}(y)} \nabla \ell_c(\theta; y, x) f_m(y, x; \theta) dx \nonumber\\
        &= \int_{\mathcal{X}(y)} \frac{f_m(y, x; \theta)}{f_c(y, x; \theta)} \nabla f_c(\theta; y, x) dx \nonumber\\
        &= \int_{\mathcal{X}(y)} \frac{1}{f(y; \theta)} \nabla f_c(\theta; y, x) dx\nonumber\\
        &= \frac{1}{f(y; \theta)} \int_{\mathcal{X}(y)} \nabla f_c(\theta; y, x) dx\nonumber\\
        &= \frac{1}{f(y; \theta)} \nabla \int_{\mathcal{X}(y)} f_c(\theta; y, x) dx \nonumber\\
        &= \frac{1}{f(y; \theta)} \nabla f(y; \theta)\nonumber\\
        &= S(\theta; y) \nonumber
    \end{align}
    Proceeding now to (\ref{eq:obs_info_identity}), we decompose the observed data log-likelihood as
    \begin{align*}
        \ell(\theta; y) &= \ell_c(\theta; y, x) - \ell_m(\theta; y, x)
    \end{align*}
    Differentiating twice and taking conditional expectations of both sides yields the required result.
\end{proof}

Note that the matrices $\mathcal{I}_c$ and $\mathcal{I}_m$ are not observed information matrices (negative Hessians), but conditional expectations of observed information matrices. 
An alternative to Proposition \ref{thm2:EM_decomp} part (\ref{eq:obs_info_identity}), which involves only conditional expectations of complete data quantities, is given in the following proposition.

\begin{proposition}[Louis' Identity]
    \label{thm2:info_decomp}
    Let $\hat{\theta}$ be a critical point of the observed data log-likelihood. Assuming that differentiation and integration can be exchanged and that all given expectations are finite, we can write the observed information of the observed data distribution at $\theta$ as
    \begin{align}
        I(\theta) = \mathcal{I}_c(\theta) - \bE_{\theta} [ S_c(\theta) S_c(\theta)^T | Y=y] + S(\theta) S(\theta)^T
    \end{align}
    In particular, if $\hat{\theta}$ is a critical point of the observed data log-likelihood, then
    \begin{align}
        I(\hat{\theta}) = \left. \left(\mathcal{I}_c(\theta)  -  \bE_{\theta} [ S_c(\theta) S_c(\theta)^T | Y=y]\right) \right|_{\theta = \hat{\theta}} \label{eq:info_at_MLE}
    \end{align}
\end{proposition}

\begin{proof}
    We follow the derivation of \citet{Lou82}. For brevity, we write $f(\theta)$ and $f_c(\theta)$ for $f(y; \theta)$ and $f(y, x; \theta)$ respectively. Consider the following two Hessians:
    \begin{align}
        \nabla^2 \ell(\theta) &= \nabla \left[ \int_{\mathcal{X}(y)} \frac{\nabla f_c(\theta) dx}{f(\theta)} \right]\\
        &= \int_{\mathcal{X}(y)} \frac{\nabla^2 f_c(\theta)}{f(\theta)} dx - \frac{1}{f(\theta)^2}\left( \int_{\mathcal{X}(y)} \nabla f_c(\theta) dx \right) \left( \int_{\mathcal{X}(y)} \nabla f_c(\theta) dx \right)^T\\
        &= \bE_\theta \left[ \left. \frac{\nabla^2 f_c(\theta)}{f_c(\theta)} \right| Y=y \right] - \bE_\theta \left[ \left. \frac{\nabla f_c(\theta)}{f_c(\theta)} \right| Y=y \right] \bE_\theta \left[ \left. \frac{\nabla f_c(\theta)}{f_c(\theta)} \right| Y=y \right]^T\\
        &= \bE_\theta \left[ \left. \frac{\nabla^2 f_c(\theta)}{f_c(\theta)} \right| Y=y \right] - S(\theta; y) S(\theta; y)^T \label{eq:hess_obs_lik}\\
        \nabla^2 \ell_c(\theta) &= \nabla \left( \frac{\nabla f_c(\theta)}{f_c(\theta)} \right)\\
        &= \frac{\nabla^2 f_c(\theta)}{f_c(\theta)} - S_c(\theta) S_c(\theta)^T \label{eq:hess_comp_lik}
    \end{align}
    Combining lines \ref{eq:hess_obs_lik} and \ref{eq:hess_comp_lik}, we get
    \begin{align}
        \nabla^2 \ell(\theta) &= \bE_\theta [ \nabla^2 \ell_c(\theta) | Y=y] + \bE_\theta [ S_c(\theta) S_c(\theta)^T | Y=y] - S(\theta; y) S(\theta; y)^T \label{eq:hess_obs_lik2}
    \end{align}
    Finally, evaluating line (\ref{eq:hess_obs_lik2}) at $\theta = \hat{\theta}$ makes the rightmost term vanish, thereby yielding the required expression.
\end{proof}

Proposition \ref{thm2:info_decomp} is known as Louis' standard error formula. Other decompositions for the observed information matrix of the observed data likelihood do exist; see, e.g., \citet{Oak99,McL08}. However, the one due to Louis will be most useful to us later.

\subsection{Example: Gene Frequency Estimation}
\label{sec:eg-genes}

In this section, we describe an analysis which will be used as a running example on which to illustrate the methods we describe. 

Consider the genetics problem of estimating allele frequencies based on observed phenotypes. Often, a single phenotype can be encoded by multiple genotypes with different configurations of dominant and recessive alleles. This is sometimes referred to as the problem of gene frequency estimation\index{Gene Frequency Estimation}. Our analysis closely follows Example 2.4 from \citet{McL08}.

We investigate a simplified model for blood type which consists of only the ABO blood group. See Chapter 5 of \citet{Dea05}, for a detailed overview of the biology of blood types\index{Blood Type}. There are three alleles for this gene: A, B and O. Allele O is recessive, while alleles A and B exhibit co-dominance. That is, genotypes AO and AA encode blood type A, genotypes BO and BB encode blood type B, genotype OO encodes blood type O, and genotype AB encodes blood type AB. Suppose that we seek to estimate the proportion of each allele within a population, based on a sample of individuals' phenotypes. \citet{Fuj78} report blood types of 4,464,349 people in Japan collected between 1964 and 1975. This sample is so large that any reasonable statistic will have sampling variability practically equal to zero. In order to retain a non-trivial level of uncertainty for the purposes of illustration, we focus on a single administrative division, Oto, in Nara Prefecture. See Figure \ref{tab2:blood_type} for details. We note here that, given the simplicity of our example, we use it to illustrate the various methods but caution readers against drawing any general conclusions for more challenging problems. In particular, we calibrate the amount of computing used by each method using general recommendations by the authors who proposed these methods (or by choosing simple numbers which might be selected for a first pass by an analyst), rather than standardizing computational effort across methods.

\begin{table}
    \centering
    \caption{Observed frequency and theoretical probability of each blood type \citep{Fuj78}}
    \begin{tabular}{c|cccc}
        Blood Type & O & A & B & AB\\
        \hline
        Random Variable & $Y_1$ & $Y_2$ & $Y_3$ & $Y_4$\\
        Observed Frequency & 10 & 16 & 7 & 1\\
        Probability & $r^2$ & $p^2 + 2pr$ & $q^2 + 2qr$ & $2pq$
    \end{tabular}
    \label{tab2:blood_type}
\end{table}

Let $Y_1$, $Y_2$, $Y_3$ and $Y_4$ be the number of people with blood type O, A, B and AB respectively, and $Y = (Y_1, Y_2, Y_3, Y_4)$. Let $r$, $p$ and $q$ be the proportions of alleles O, A and B respectively within the population of interest. Since $r + p + q = 1$, we let $\theta = (p, q)$ be our target of inference. Pretending that the population size is fixed and that simple random sampling was employed, $Y$ follows a multinomial distribution with $n = 34$. Assuming homogeneous genetic mixing, the probability vector for $Y$ is $\pi = (r^2, p^2 + 2pr, q^2 + 2qr, 2pq)$, where we retain $r$ in our notation as shorthand for $1 - p - q$.

Maximizing the likelihood in this model involves solving the score equations, a system of two 3rd-degree polynomials in $p$ and $q$. This can be done numerically, and gives estimates $p = 0.299$ and $q = 0.128$. These values match the ones given by \citet{Fuj78}. The information matrix and asymptotic covariance (inverse information matrix) are given by
\begin{align}
    I(\hat{\theta}) &= \begin{bmatrix}
        276 & 84.8\\
        84.8 & 584
    \end{bmatrix}\\
    \hat{\Sigma}_\mathrm{MLE} &= \begin{bmatrix}
        3.79 \cdot 10^{-3} & -5.49 \cdot 10^{-4}\\
        -5.49 \cdot 10^{-4} & 1.79 \cdot 10^{-3}
    \end{bmatrix} \label{eq:obs_lik_SE}
\end{align}
The asymptotic standard errors for our estimators are thus $\sqrt{3.79 \cdot 10^{-3}} = 0.062$ and $\sqrt{1.79 \cdot 10^{-3}} = 0.042$ for $\hat{p}$ and $\hat{q}$ respectively.

\subsubsection{Complete Data}

The problem of gene frequency estimation would be much simpler if we could also observe individuals' genotypes. As such, consider augmenting the observed data $Y$ by further classifying individuals by genotype. Let $X = (X_1, \ldots, X_6)$ be the genotypes of the individuals represented in Table \ref{tab2:blood_type}. See Table \ref{tab2:blood_type_complete}.

\begin{table}
    \centering
    \caption{Terminology and probabilities for our augmented version of the dataset in \citet{Fuj78}. We also give the blood type coded for be each genotype.}
    \begin{tabular}{c|cccccc}
        Genotype & OO & AO & AA & BO & BB & AB\\
        \hline
        Random Variable & $X_1$ & $X_2$ & $X_3$ & $X_4$ & $X_5$ & $X_6$\\
        Probability & $r^2$ & $2pr$ & $p^2$ & $2qr$ & $q^2$ & $2pq$\\
        Blood Type & O & A & A & B & B & AB
    \end{tabular}
    
    \label{tab2:blood_type_complete}
\end{table}

Note that we can express the elements of $Y$ in terms of $X$. Specifically, $Y_1 = X_1$, $Y_2 = X_2 + X_3$, $Y_3 = X_4 + X_5$ and $Y_4 = X_6$. This corresponds to summing components of $X$ which correspond to the same blood type. The distribution of $X$ is multinomial, with the same sample size as $Y$, and probability vector given in Table \ref{tab2:blood_type_complete}.

See Appendix \ref{app:blood_complete} for the complete data likelihood function and its derivatives.

\subsubsection{EM Algorithm}

The gene frequency estimation problem fits nicely into the EM algorithm framework. In this section, we present key quantities and results of our analysis. See Appendix \ref{app:blood}, especially parts \ref{app:blood_miss} and \ref{app:EM}, for more details.

The EM objective function at iteration $k$ is
\begin{align}
    Q(\theta|\theta_{k-1}) &\equiv  \bE_{\theta_{k-1}}(n_O | y) \log r + \bE_{\theta_{k-1}}(n_A | y) \log p + \bE_{\theta_{k-1}}(n_B | y) \log q\\
    &=: \nu^{(k-1)}_O \log r + \nu^{(k-1)}_A \log p + \nu^{(k-1)}_B \log q
\end{align}
where $\nu^{(k-1)}_O$, $\nu^{(k-1)}_A$ and $\nu^{(k-1)}_B$ are the expected number of O, A and B alleles respectively given $Y=y$ and $\theta = \theta_{k-1}$. See Appendix \ref{app:blood_miss} for explicit formulas. 

Starting with $\theta_0 = (1/3, 1/3)$ corresponding to equal proportions of the three alleles, Figure \ref{fig:blood_EM_traj} gives trajectories for the EM estimates of $p$ and $q$ using the data in Figure \ref{tab2:blood_type}. These estimates converge quite quickly to the maximizer of the observed data likelihood, given by the horizontal dashed lines.
\begin{figure}
    \centering
    \caption{Trajectory of EM estimates for $p$ and $q$ for the blood type example. Horizontal dashed lines give the values of the MLE.}
    \label{fig:blood_EM_traj}
    \includegraphics[width=0.9\textwidth]{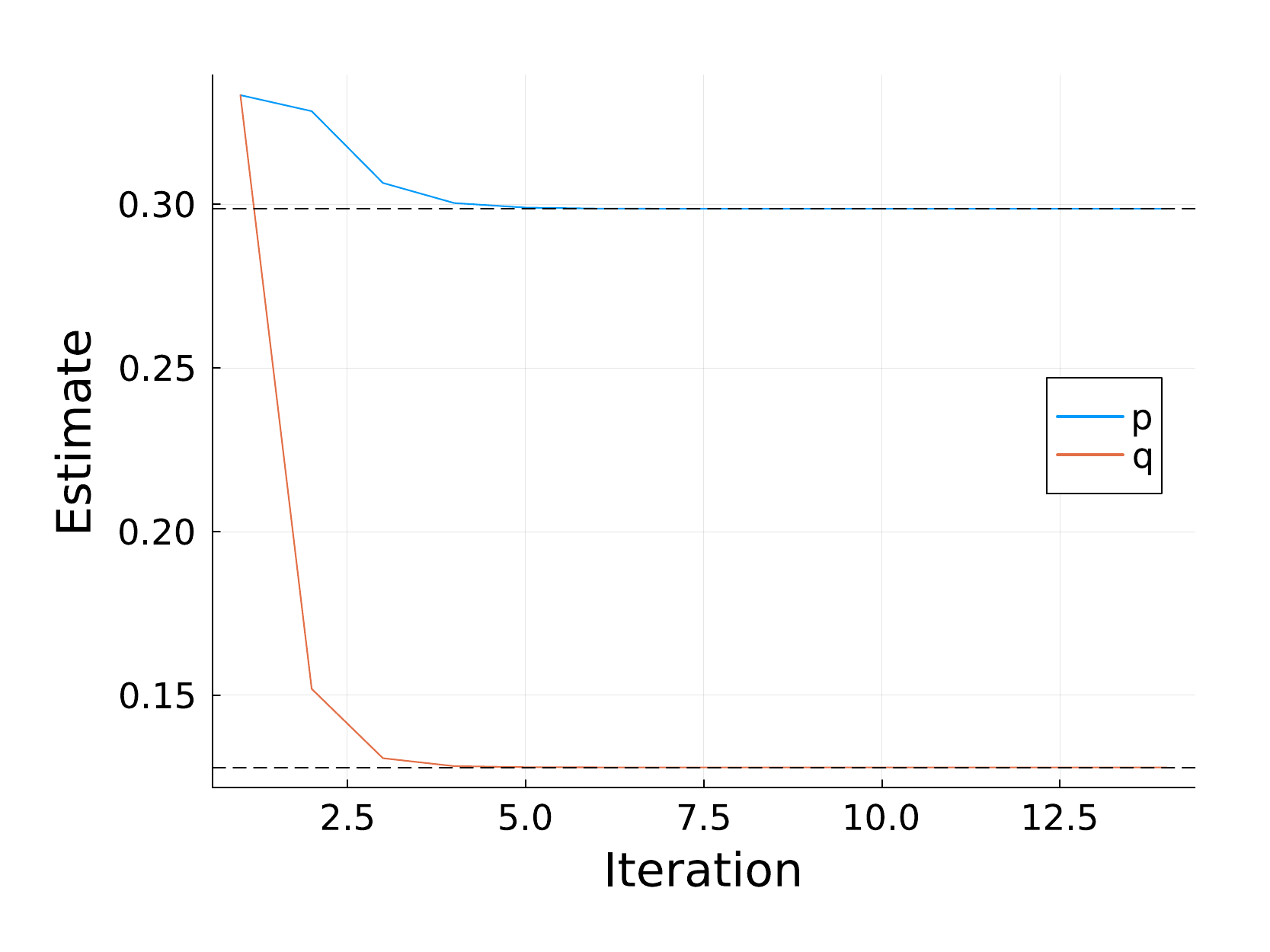}
\end{figure}
Beyond computing the observed data MLE, we also need the standard error of our estimator. To this end, we compute the observed data information matrix using Louis' Method (see Proposition \ref{thm2:info_decomp}). The asymptotic covariance matrix of our MLE is approximately the inverse of this information matrix. Omitting details (see Appendices \ref{app:blood_complete} and \ref{app:blood_miss}), both the observed data information matrix and asymptotic covariance match those obtained from the observed data likelihood. 

\section{The Monte Carlo EM Algorithm}
\label{sec:MCEM}

The Monte Carlo EM, or MCEM, algorithm was first proposed by \citet{Wei90} (see also Section 4.5 of \citealp{Tan96}, for a textbook-level treatment of the basic idea with numerous examples). This method proceeds by replacing the conditional expectation in the E-step of the EM algorithm with a Monte Carlo average. More precisely, at each iteration we generate observations from the missing data distribution (i.e.\ the conditional distribution of the missing data given the observed data), and average the complete data likelihood over this Monte Carlo sample. Formally, at a given iteration of the MCEM algorithm, let $X_1,\ldots, X_M$ be a Monte Carlo sample (not necessarily iid) from the distribution of $X|Y=y$ with $\theta$ set to some value, say $\theta_0$. Write
\begin{align}
    \hat{Q}(\theta|\theta_0) &= \sum_{i=1}^M w_i \ell_c(\theta; y, X_i) \label{eq:MCEM_objective}
\end{align}
where the $w_i$ are sampling weights. We write $\hat{\bE} \phi$ \index{$\hat{\bE}$} for the estimated mean of a function $\phi$ based on a weighted sample (for iid sampling, take all weights equal to $1/M$), so $\hat{Q}(\theta|\theta_0)$ can be re-written as $\hat{\bE} \ell_c (\theta; y, X)$. In this section, we focus only on iid sampling, but see Section \ref{sec:simulation} for discussions of some alternative sampling methods. The estimate of $\theta$ at iteration $k$ is then the maximizer of the MCEM objective function: $\hat{\theta} = \argmax_\theta \hat{Q}(\theta|\theta_{k-1})$. Write $\hat{\theta}_{k}$ for the MCEM estimate at iteration $k$. The next iteration then proceeds with maximizing $\hat{Q}$ with $\theta_0 = \hat{\theta}_k$.  

Provided that the MCEM algorithm converges to a critical point of the observed data likelihood, we can use Proposition \ref{thm2:info_decomp} to estimate the observed data information matrix. Specifically, after declaring convergence, we generate a new Monte Carlo sample and use it to approximate the conditional expectations in equation (\ref{eq:info_at_MLE}).

The MCEM algorithm has the advantage of circumventing the challenge of computing potentially intractable conditional expectations for the EM algorithm. However, this analytical simplification comes at the cost of introducing some new computational difficulties. In this section, we outline the main problems faced by the MCEM algorithm and present various solutions which have been proposed in the literature. We focus primarily on practical aspects of the MCEM algorithm; see \citet{For03} for a thorough analysis of the convergence properties of MCEM, and \citet{Nea13} for a review of this and other convergence theory for MCEM. Furthermore, \citet{For03} also suggest an offline averaging procedure, where a trajectory of estimates is replaced by its cumulative means \citep[see also Section 11.1.2.2 of][]{Cap05}. They prove that this average sequence converges faster than the original trajectory, provided that certain conditions are met for the growth rate of the Monte Carlo sample sizes. We do not explore this averaging process any further.

Two connected problems which have received considerable attention in the literature are how to choose the Monte Carlo sample size at each iteration, and when to terminate the MCEM algorithm. These were identified early by \citet{Wei90}, but did not receive systematic treatment until later. \citet{Wei90} suggest examining a plot of the parameter estimates across iterations, and either terminating or increasing the Monte Carlo size (i.e. Monte Carlo sample size) when the plot appears to stabilize. \citet{Cha95} use a pilot study to choose the Monte Carlo sample size, and terminate when a confidence interval for the improvement in the observed data log-likelihood between successive iterations contains zero. \citet{Boo99} frame each MCEM iteration as an M-estimation problem for estimating the deterministic EM update. They increase the Monte Carlo size if an asymptotic confidence interval for the EM update contains the previous iteration's parameter estimate, and terminate when multiple consecutive iterations' estimates have sufficiently small relative error. \citet{Caf05} build confidence bounds for the increment in the EM objective function at each iteration of the MCEM algorithm. They increase the Monte Carlo size at each iteration until the lower bound is positive and terminate when the upper bound is sufficiently small.

In the rest of this section, we give more detail on each of the implementations introduced above. We also illustrate each method on the blood type dataset described in Section \ref{sec:eg-genes}. The relevant conditional distribution and likelihood calculations are described in Appendix \ref{app:blood}.

\subsection{Early Work (\citealp{Wei90})}

In their seminal work, \citet{Wei90} propose the MCEM algorithm and present a simple implementation. They illustrate that the complete data gradient and Hessian are easily obtained at each iteration from the Monte Carlo sample and, following \citet{Lou82}, give an estimator for the observed data information matrix. Regarding convergence, \citeauthor{Wei90} recommend plotting the parameter estimates across iterations and stopping when the estimates appear to stabilize around some constant. When this stabilization is detected, one can either declare convergence, or increase the Monte Carlo size and continue iterating until the estimate trajectory again stabilizes.

In order to apply the MCEM algorithm to estimate allele frequencies in the blood type problem, we must specify the number of iterations, $K$, and the Monte Carlo size for each iteration, $M$. Starting conservatively, we use $K=50$ and $M=100$. Figure \ref{fig:blood_naive_MCEM_traj1} gives trajectories of the MCEM estimates of $p$ and $q$. These estimates appear to converge quickly to a stationary mean, but there is still some uncertainty around this mean. As such, we run MCEM for another 20 iterations with $M=1000$, staring with the final value from our first run. See Figure \ref{fig:blood_naive_MCEM_traj2}. The trajectories from our second run are much more stable around their means. We use the final values from these trajectories as our estimates: $\hat{p} = 0.298$ and $\hat{q} = 0.128$. These values closely match the maximizer of the observed data likelihood. 

As can be seen in Section \ref{sec:eg-genes}, the standard errors of our estimators are on the order of $10^{-2}$. This is much larger than the Monte Carlo variability seen in either plot of Figure \ref{fig:blood_naive_MCEM_traj}.

   

\begin{figure}
    \centering
    \caption{Trajectory of MCEM estimates of $p$ and $q$ for the blood type example. The horizontal lines correspond to maximum likelihood estimates.}
    \label{fig:blood_naive_MCEM_traj}
    \subfloat[$M=100$]{\includegraphics[width = 0.75\textwidth]{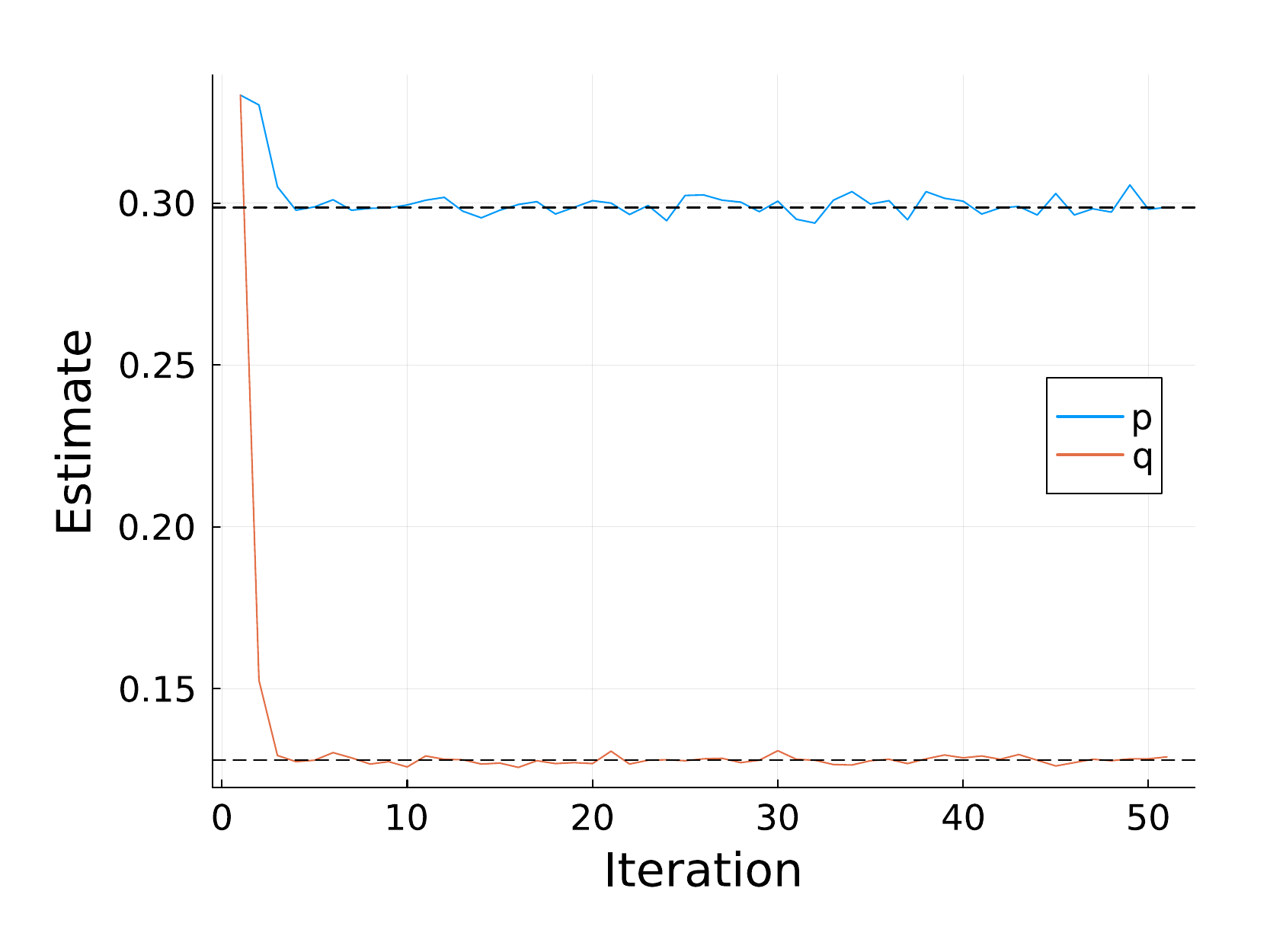} \label{fig:blood_naive_MCEM_traj1}}
    
    \subfloat[$M=1000$]{\includegraphics[width = 0.75\textwidth]{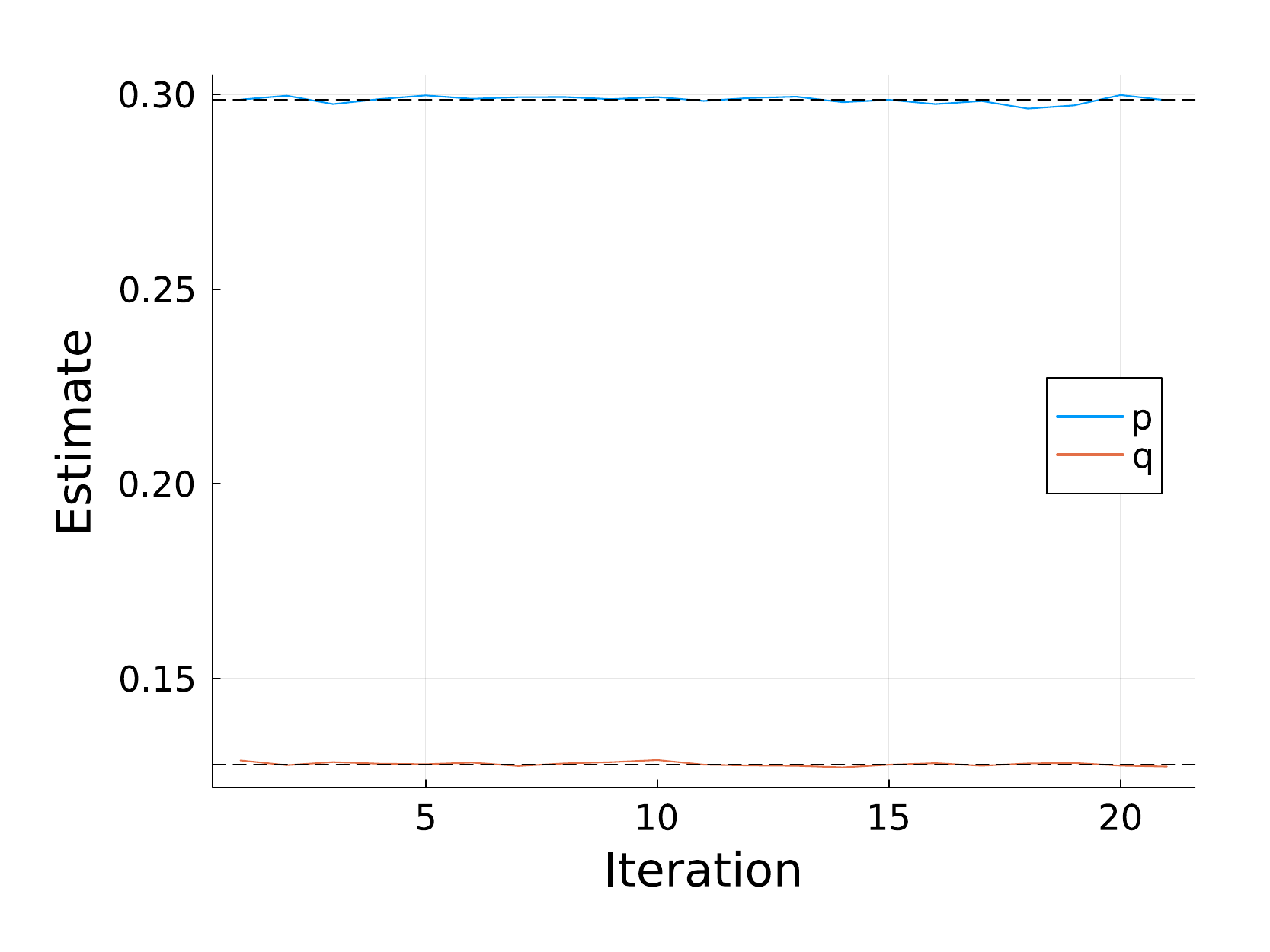} \label{fig:blood_naive_MCEM_traj2}}   
\end{figure}

\subsection{Running a Pilot Study \citep{Cha95}}

Building on the ideas of \citeauthor{Wei90}, \citet{Cha95} develop a method for both choosing the Monte Carlo size and deciding when to terminate the MCEM algorithm. The method of \citeauthor{Cha95} includes numerous choices for which they do not give specific guidance. We describe the procedure in general terms, but details for any particular implementation will need to be explored in the context of the dataset being analysed. 

The algorithm presented by \citeauthor{Cha95} is based on an identity which allows us to estimate observed data likelihood ratios by Monte Carlo averages of complete data likelihood ratios. More precisely, we write
\begin{align}
	\frac{\mathcal{L}(\theta_1; y)}{\mathcal{L}(\theta_2; y)} &= \bE_{\theta_2} \left[ \left. \frac{\mathcal{L}_c(\theta_1; y, X)}{\mathcal{L}_c(\theta_2; y, X)} \right| Y=y \right] \label{eq:Cha95}
\end{align}
See \citeauthor{Cha95} for a derivation. To apply equation (\ref{eq:Cha95}) to the MCEM algorithm, we replace the conditional expectation on the right-hand side with a Monte Carlo average from the corresponding conditional distribution. This adjustment allows us to estimate log-likelihood ratios from the observed data distribution, without ever directly evaluating the observed data likelihood.

The algorithm proposed by \citeauthor{Cha95} can be divided into two parts. The first part is a pilot study, in which we compute a standard error for our log-likelihood ratio estimator near the MLE, and determine what Monte Carlo size is required to get this standard error below a pre-specified threshold. Note that this standard error quantifies Monte Carlo uncertainty, not sampling variability of the observed data. In the second part, we increase the Monte Carlo size to get confidence intervals narrower than some pre-specified threshold, and continue iterating until a confidence interval for the true observed data log-likelihood ratio contains zero. This two-part procedure reflects the suggestion of \citet{Wei90} to run MCEM until it appears to stabilize, then increase the Monte Carlo sample size to get a more precise estimate.

The pilot study portion of \citeauthor{Cha95}'s method consists of running the MCEM algorithm with a fixed, ``moderately large'' Monte Carlo size (although they give no general guidance on exactly how large this should be).  In addition to tracking the parameter estimates across iterations, we also record the estimated log-likelihood ratio of the current estimate relative to the starting point of the algorithm. This ratio is computed by keeping a running cumulative sum of all one-step log-likelihood ratios. We terminate our pilot study after a pre-specified number of iterations, and identify the estimate with largest estimated observed data log-likelihood ratio. 

After concluding our pilot study, we select a few estimates from iterations near the maximizer (\citeauthor{Cha95} suggest the 10 which follow the maximizer but give no justification for this number). All the selected estimates are thought of as approximately equivalent to the maximizer. We then perform a few single-iteration MCEM runs from each of the chosen estimates and pool information about the variability of a one-step change in the estimated observed data log-likelihood. See Section 2.3 of \citet{Cha95} for a more detailed description of their pilot study and variance estimator. Once we have an estimate of the Monte Carlo variance of our log-likelihood ratio estimator, we calculate what Monte Carlo sample size would be required to get that variance below some pre-specified threshold\footnotemark. We then use this newly determined Monte Carlo size for a follow-up MCEM run.

\footnotetext{Provided that the one-step observed data likelihood ratio is evaluated close to the MLE, \citeauthor{Cha95} show that its Monte Carlo standard error scales like $1/M$ rather than the usual $1 / \sqrt{M}$, where $M$ is the Monte Carlo size.}

For our second run of MCEM, we return to the optimal parameter estimate from our pilot study and, using our newly determined Monte Carlo size, continue iterating the MCEM algorithm. At each step now, we also construct a confidence interval for the true observed data log-likelihood ratio corresponding to the current parameter update (i.e., between two consecutive parameter values, not from the starting point to the current estimate). We terminate the algorithm when such a confidence interval contains zero. This corresponds to no evidence of an improvement in the observed data likelihood.

Note that the estimated observed data log-likelihood ratio computed at each iteration is by a Monte Carlo average. In order to avoid bias, we do not recycle the current iteration's Monte Carlo sample to estimate this ratio. Instead, we generate the Monte Carlo sample which will be used in the next iteration, and use this new sample to estimate the log-likelihood ratio\footnotemark. We are then free to use this new Monte Carlo sample to compute the next iteration's parameter update. 

\footnotetext{Due to our Monte Carlo sample using the new parameter estimate rather than the old, we must actually estimate the reciprocal of the likelihood ratio that we want, then multiply its logarithm by $-1$. This is reflected in the formulas of \citeauthor{Cha95} but not discussed explicitly.}

Although a superficial reading of \citeauthor{Cha95} suggests that their method is quite complicated, its implementation on our blood type dataset is fairly straightforward. Applying this algorithm to our example, we get the parameter estimate trajectory shown in Figure \ref{fig:blood_CL_MCEM_estimates}. Figure \ref{fig:blood_CL_MCEM_likelihoods} gives the trajectory of estimated observed data log-likelihood ratios relative to the starting point of the algorithm (starting with $\hat{p} = \hat{q} = 1/3$), along with pointwise $95\%$ Wald-type confidence bands for the post-pilot study iterations. Recall that these confidence bands are used to assess convergence of the algorithm. We require that the standard error of our estimated log-likelihood increment be at most $10^{-3}$, since this is much smaller than the statistical uncertainty given in Section \ref{sec:eg-genes}. Note that this standard error only applies near the maximizer of the observed data MLE, so we do not report confidence intervals until the second stage of our MCEM run. We see in Figure \ref{fig:blood_CL_MCEM_estimates} that the Monte Carlo fluctuations are much smaller than the statistical fluctuations in our problem. Our estimated parameter values are $\hat{p} = 0.298$ and $\hat{q} = 0.129$, which closely match the MLE.

\begin{figure}
    \centering
    \caption{MCEM estimates of $p$ and $q$ for the blood type example, based on the method of \citet{Cha95}. The vertical line shows the end of the pilot study. The horizontal dashed lines correspond to maximum likelihood estimates.}
    \includegraphics[width = 0.75\textwidth]{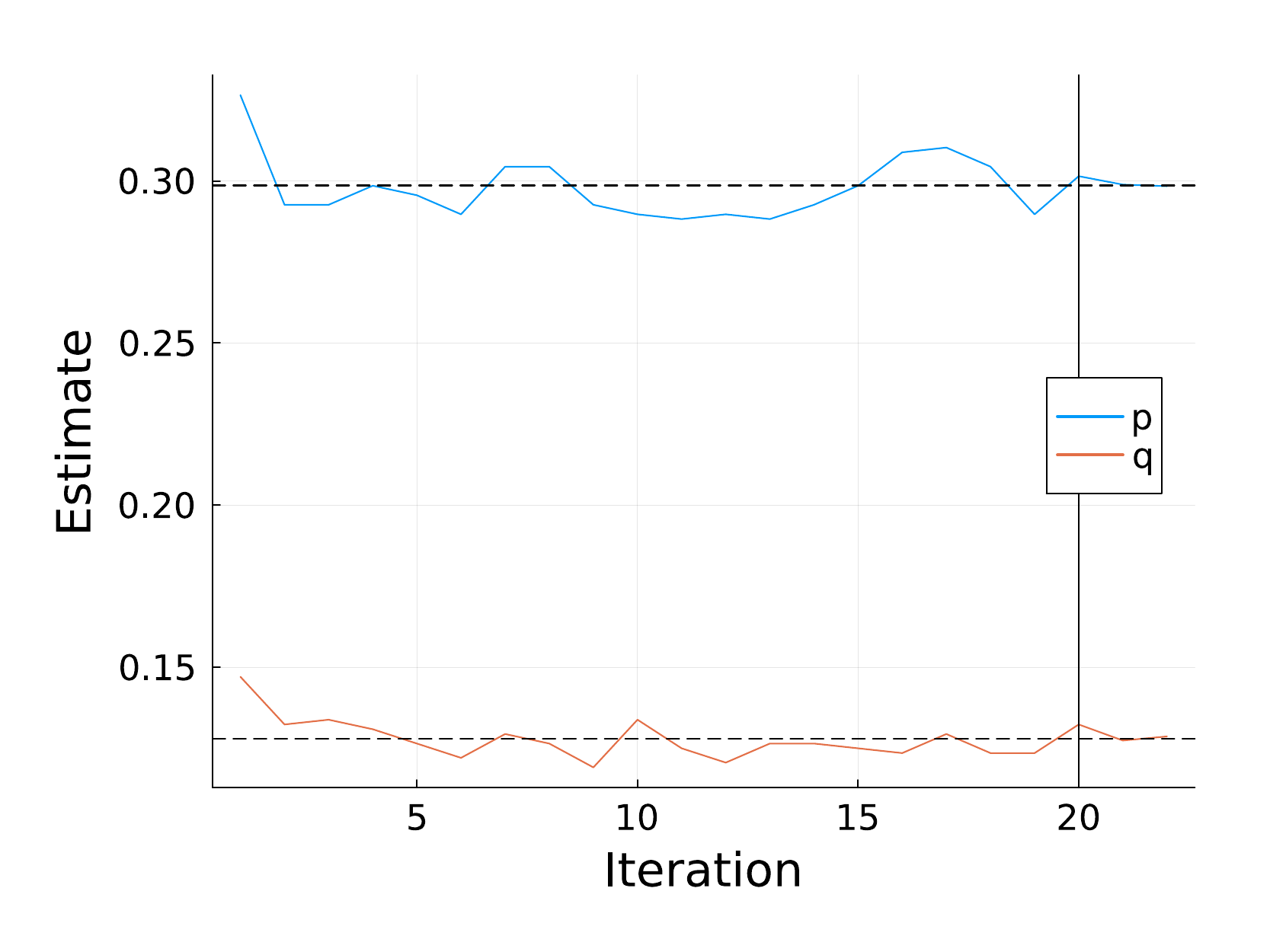} 
    \label{fig:blood_CL_MCEM_estimates}
\end{figure}

\begin{figure}
    \centering
    \caption{Estimated observed data log-likelihood ratio, based on the method of \citet{Cha95}. Red dashed lines give $95\%$ pointwise confidence bands, and the vertical line shows the end of the pilot study.}
    \subfloat[Full trajectory]{\includegraphics[width = 0.75\textwidth]{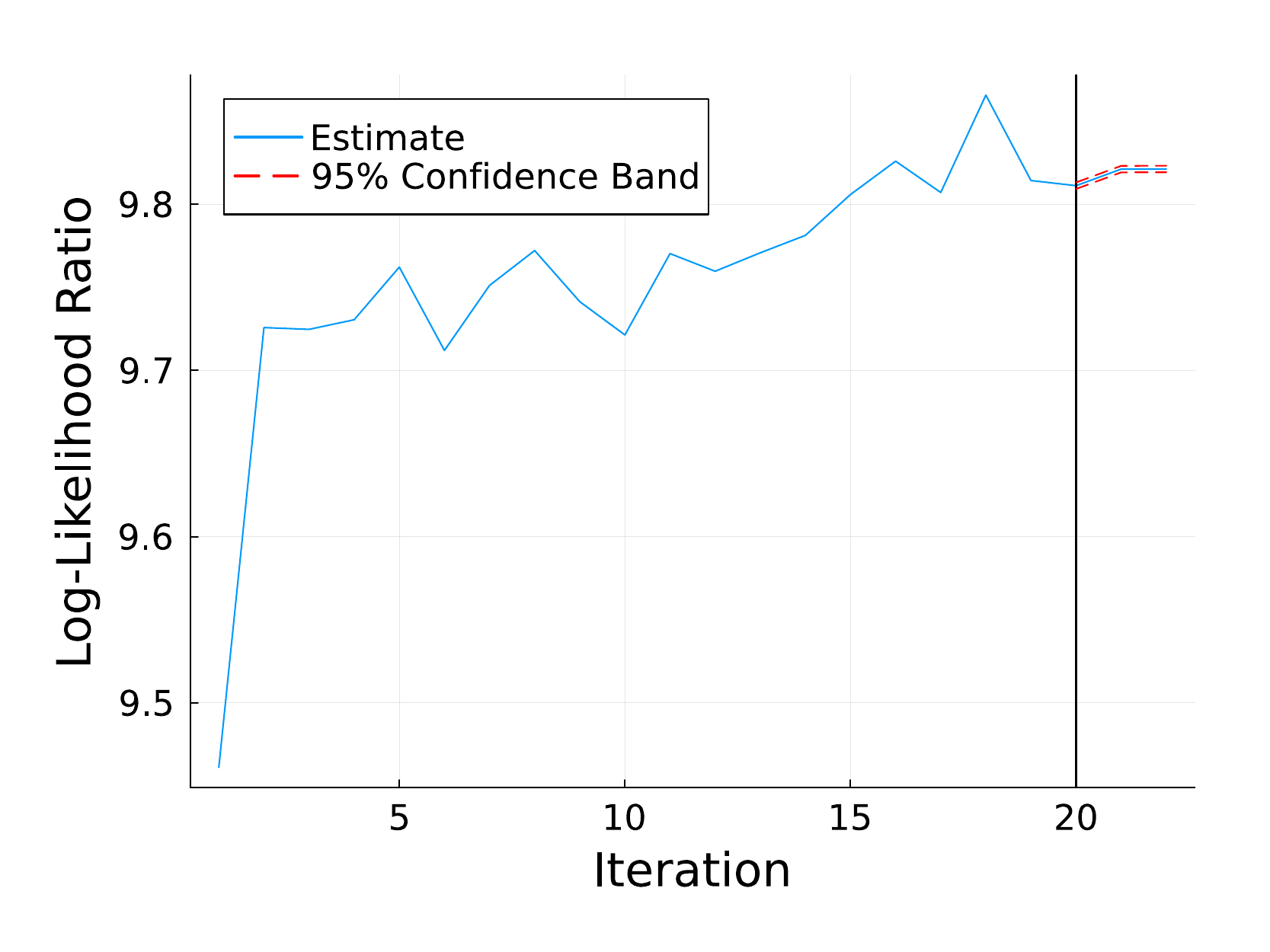} \label{fig:blood_CL_MCEM_likelihoods_full}}

    \subfloat[Post-pilot study iterations only]{\includegraphics[width = 0.75\textwidth]{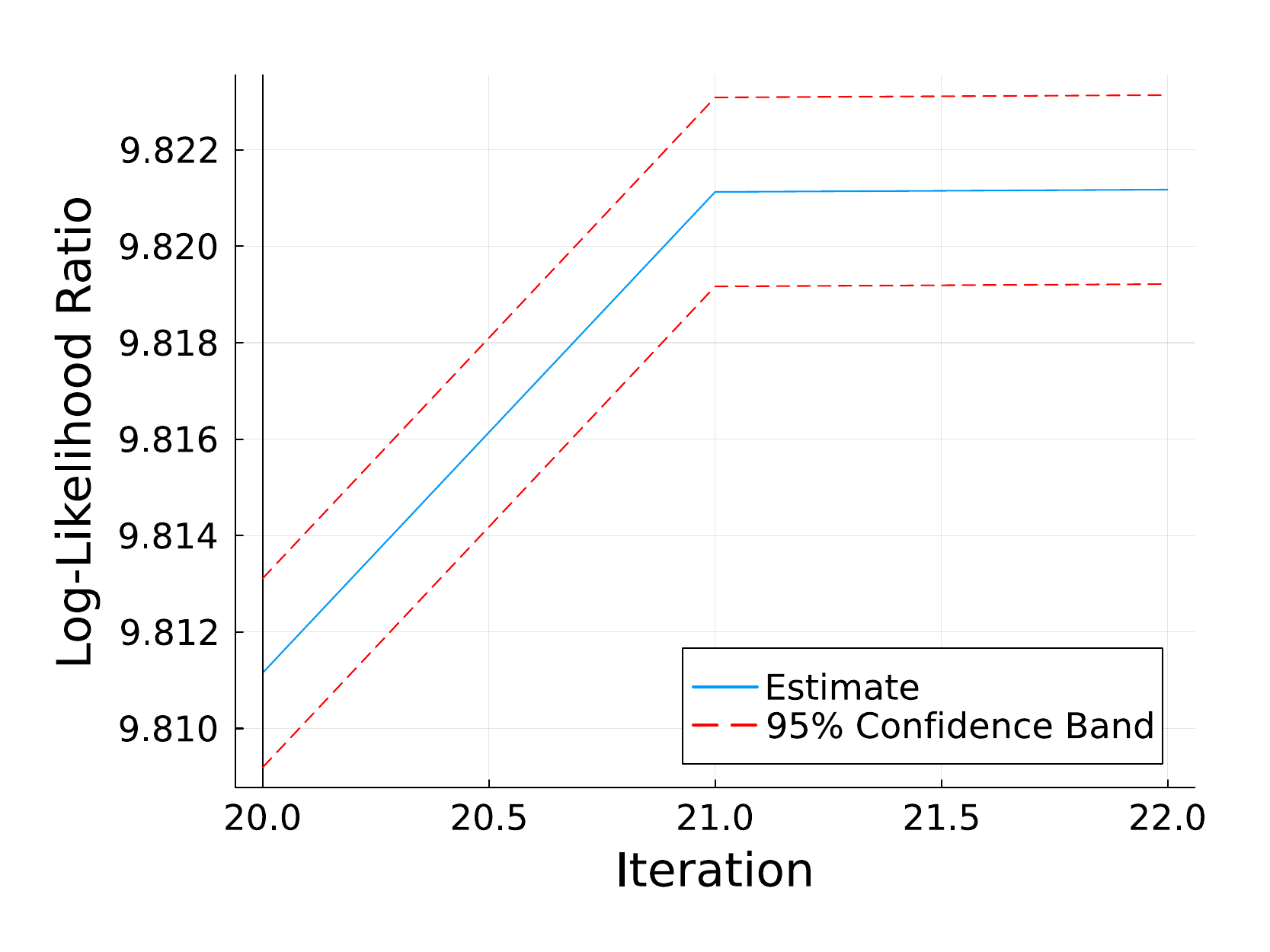} \label{fig:blood_CL_MCEM_likelihoods_zoomed}}   
    
    \label{fig:blood_CL_MCEM_likelihoods}
\end{figure}

    

\subsection{Uncertainty Quantification for the Parameter Estimate (\citealp{Boo99})}

 \citet{Boo99} use a somewhat different approach from either \citet{Wei90} or \citet{Cha95} to understand the behaviour of the MCEM algorithm. The method of \citeauthor{Boo99} is based on quantifying Monte Carlo uncertainty of the MCEM update as an approximation to the update which would have been made by the deterministic EM algorithm from the same starting point. They recommend starting the MCEM algorithm with a small Monte Carlo size, and adding more observations only when the parameter estimates are no longer changing discernibly across iterations. More formally, \citeauthor{Boo99} suggest building a confidence interval for the EM update based on the Monte Carlo variability of the MCEM update at each iteration. If this interval contains the previous iteration's parameter estimate, then the Monte Carlo variability is too large relative to the size of the parameter updates, in which case more samples are required. The authors then recommend assessing convergence by checking for small relative error in the parameter updates. To account for the possibility of Monte Carlo variability leading to two consecutive estimates being similar before the algorithm has `converged', they suggest stopping only after the relative error is small for three consecutive iterations.

The confidence interval used to quantify Monte Carlo uncertainty within an iteration is obtained by framing each step of the MCEM algorithm as the solution of an M-estimation problem. This allows us to inherit the desirable properties of M-estimators; specifically, asymptotic normality (see, e.g. \citealp{van98}). Following the usual M-estimator construction and assuming that the relevant regularity conditions hold, we can estimate the asymptotic variance of the MCEM parameter estimator at each iteration. Note that this standard error applies to Monte Carlo variability within an iteration; it does not measure sampling variability due to the observed data. 

More formally, If we write $\tilde{\theta}_k$ for the EM update based on $\hat{\theta}_{k-1}$ then, assuming sufficient smoothness and moment conditions, we get the following expression for the MCEM update:
\begin{align}
    \sqrt{M_k}(\hat{\theta}_k - \tilde{\theta}_k) &= - \sqrt{M_k} \left[ \nabla^2 Q(\tilde{\theta}_k|\hat{\theta}_{k-1})\right]^{-1} \left[\nabla \hq(\tilde{\theta}_k|\hat{\theta}_{k-1}) \right] + o_p(1) \label{eq:th_as_dist}
\end{align}
as $M_k \rightarrow \infty$, where $M_k$ is the Monte Carlo size used to compute $\hat{\theta}_k$, $\nabla$ denotes differentiation with respect to the left argument of $Q$ or $\hat{Q}$ and $o_p(1)$ is a sequence which converges in probability to zero. Note that $\hat{\theta}_{k-1}$ is held fixed (analysis of an MCEM update is done conditional on the previous iteration's estimate). 

The first expression on the right-hand side of (\ref{eq:th_as_dist}) is the inverse Hessian of the EM objective function (fixed) while the second is the gradient of the MCEM objective function (an average). Thus, $\hat{\theta}_k$ is asymptotically normal with asymptotic covariance
\begin{align}
    &\left[ \nabla^2 Q(\tilde{\theta}_k|\hat{\theta}_{k-1})\right]^{-1} \bV \left[ S_c (\tilde{\theta}_k) | Y=y \right] \left[ \nabla^2 Q(\tilde{\theta}_k|\hat{\theta}_{k-1})\right]^{-1}\\
    &\approx \left[ \nabla^2 \hat{Q}(\hat{\theta}_k|\hat{\theta}_{k-1})\right]^{-1} \hat{\bE} \left[ S_c(\hat{\theta}_k) S_c(\hat{\theta}_k)^T | Y=y \right] \left[ \nabla^2 \hat{Q}(\hat{\theta}_k|\hat{\theta}_{k-1})\right]^{-1} 
\end{align}
where $S_c$ is the complete data score vector, and $\hat{\bE}$ is the Monte Carlo average over the missing data with $\hat{\theta}_k$ held fixed. Note that there is no first moment term in the conditional variance of $S_c$ because $\hat{\theta}_k$ is a maximizer of $\hat{\bE} [\ell_c (\theta) |Y=y]$.

Based on the above discussion, we can build an asymptotic confidence interval for $\tilde{\theta}_k$. Recall that \citeauthor{Boo99} recommend checking whether this interval contains $\hat{\theta}_{k-1}$ and, if so, increasing the Monte Carlo size for the next iteration. Specifically, they suggest starting the next iteration with a sample of size $M_{k+1} = M_k ( 1 + 1/r)$, with $r = 3,4$ or $5$ working well in their examples.

To assess convergence of the MCEM algorithm, \citeauthor{Boo99} present two criteria. The first is a familiar measure of relative error for parameter estimates between consecutive iterations:
\begin{align}
    \max_j \left( \frac{\left| \hat{\theta}_{k, j} - \hat{\theta}_{k-1,j} \right|}{\left| \hat{\theta}_{k-1,j} \right| + \delta_1} \right) < \delta_2 \label{eq:Boo99_tol}
\end{align}
where $\delta_1$ and $\delta_2$ are small positive constants, and the subscript $j$ ranges over components of $\theta$. \citeauthor{Boo99} suggest using $\delta_1 = 10^{-3}$ and $\delta_2$ between $2 \cdot 10^{-3}$ and $5 \cdot 10^{-3}$. See p.\ 436 of \citet{Sea06} and the references therein for a discussion of the form of equation (\ref{eq:Boo99_tol}). Recall that \citeauthor{Boo99} suggest terminating only when condition (\ref{eq:Boo99_tol}) is satisfied for three consecutive iterations.

Alternatively, since \citeauthor{Boo99} apply their method to the analysis of generalized linear mixed models, where pathologies may arise due to parameter estimates being too close to a boundary, they propose a second stopping rule:
\begin{align}
    \max_j \left( \frac{\left| \hat{\theta}_{k, j} - \hat{\theta}_{k-1,j} \right|}{\mathrm{SE}\left(\hat{\theta}_{k,j}\right) + \delta'_1} \right) < \delta'_2 \label{eq:Boo99_tol2}
\end{align}
where $\delta_1'$ and $\delta_2'$ are tolerances which may or may not differ from $\delta_1$ and $\delta_2$, and $\mathrm{SE}\left(\hat{\theta}_{k,j}\right)$ is the standard error at iteration $k$. The purpose of condition (\ref{eq:Boo99_tol2}) is to detect when estimated variance components are very close to zero, whereupon the numerical precision needed to satisfy condition (\ref{eq:Boo99_tol}) requires a prohibitive amount of computation.

\citet{Rip02} propose a modification to the method of \citeauthor{Boo99}. This new version uses the same stopping rule, but also uses the relative step lengths in equation \ref{eq:Boo99_tol} to determine when to increase the Monte Carlo size. Specifically, \citeauthor{Rip02} suggest computing the coefficient of variation (i.e., standard deviation divided by mean) of the relative step lengths (left-hand side of equation \ref{eq:Boo99_tol}) from the previous three iterations. If this coefficient of variation is larger than the one from the previous iteration, then we increase the Monte Carlo size for our next iteration. That is, when deciding whether to increase the Monte Carlo size at, e.g., iteration 6, we compute the coefficient of variation of the relative step lengths at iterations 3, 4 and 5. We then compare this to the coefficient of variation obtained from relative step lengths at iterations 2, 3 and 4. If the former coefficient of variation is larger, then we must increase the Monte Carlo size for iteration 6.

We apply the method of \citeauthor{Boo99} to our blood type example, with the settings recommended in their paper. Specifically, they suggest setting $\alpha = 0.25$, $k = 3$ (alternatively, $4$ or $5$), $\delta_1 = 0.001$, and $\delta_2 = 0.002$ (or as high as $0.005$). We also start with a Monte Carlo size of 10. Figure \ref{fig:blood_BH_traj} gives trajectories of the MCEM estimates, as well as the Monte Carlo size used to obtain each of these estimates. Note how the trajectories stabilize around the MLE as MC size increases. The final estimate from this method is $\hat{p} = 0.299$, $\hat{q} = 0.128$ which is very close to the MLE. The Monte Carlo fluctuations seen in Figure \ref{fig:blood_BH_traj} are also much smaller than the statistical uncertainty given in equation (\ref{eq:obs_lik_SE}).

\begin{figure}
    \centering
    \caption{Trajectory of estimates for $p$ and $q$, as well as Monte Carlo sample sizes, from the method of \citeauthor{Boo99}. Horizontal dashed lines give the maximum likelihood estimates.}
    \label{fig:blood_BH_traj}
    \includegraphics[width=0.75\textwidth]{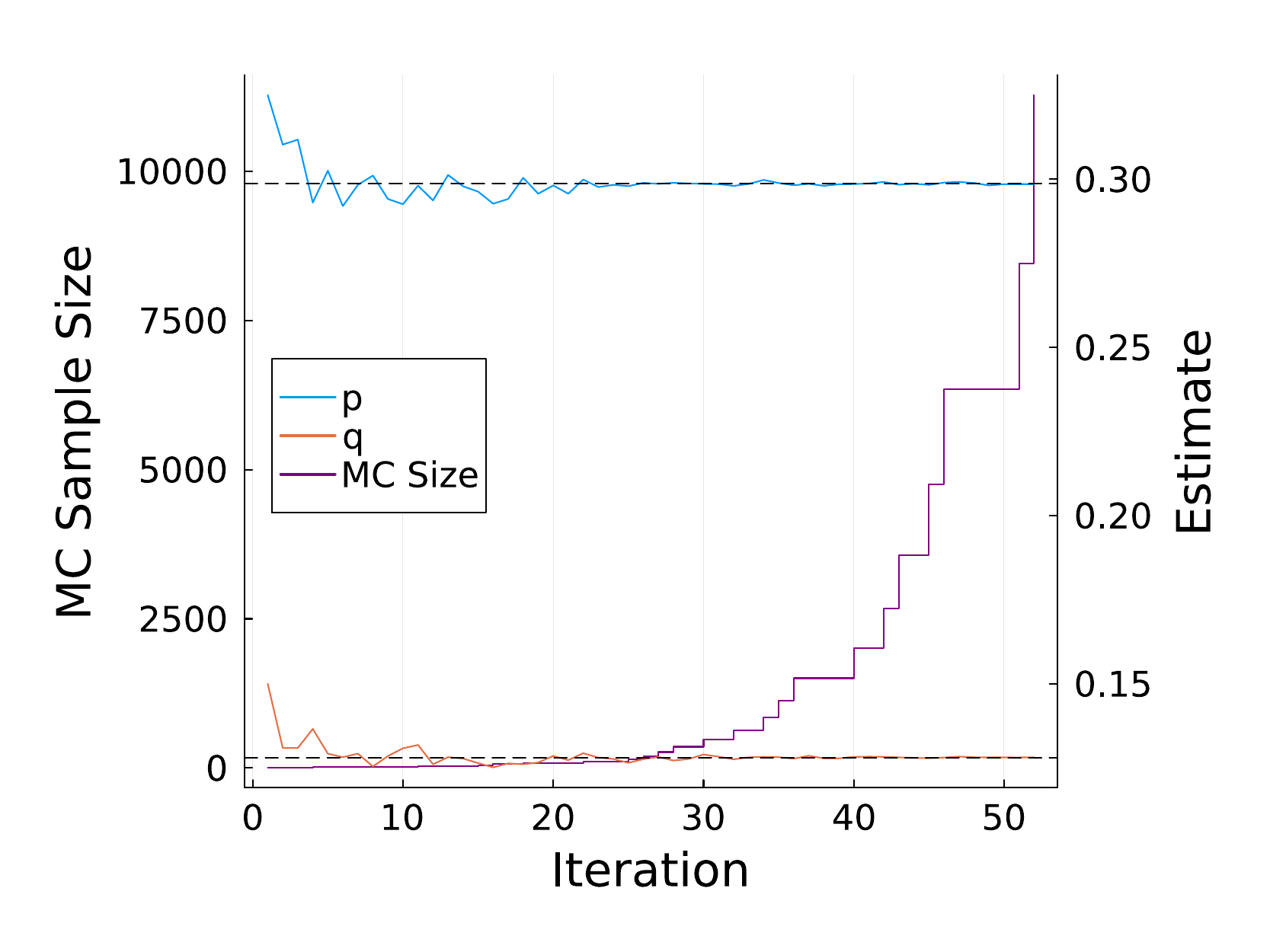}   
\end{figure}

\subsection{Uncertainty Quantification for the Objective Function \citep{Caf05}}
\label{sec:AMCEM}

The approach of \citet{Caf05} is similar in spirit to that of \citet{Boo99}. Both methods quantify Monte Carlo uncertainty in the MCEM algorithm as an approximation to the EM algorithm. The difference is that where \citeauthor{Boo99} measure uncertainty in the parameter estimates, \citeauthor{Caf05} focus on uncertainty in the objective function. Specifically, \citeauthor{Caf05} base their analysis on asymptotic normality of the MCEM increment using the following:
\begin{proposition}
    \label{thm2:Caf_normality}
    Let $\Delta \hat{Q}(\hat{\theta}_k|\hat{\theta}_{k-1}) = \hat{Q}(\hat{\theta}_{k-1}|\hat{\theta}_{k-1}) - \hat{Q}(\hat{\theta}_k|\hat{\theta}_{k-1})$. Define $\Delta Q(\hat{\theta}_k|\hat{\theta}_{k-1})$ similarly. Let $M_k$ be the Monte Carlo size at iteration $k$. Then
    \begin{align}
        \sqrt{M_k} \left[ \Delta \hat{Q}(\hat{\theta}_k|\hat{\theta}_{k-1}) - \Delta Q(\hat{\theta}_k|\hat{\theta}_{k-1}) \right] \rightsquigarrow N(0, \Sigma_k)
    \end{align}
    as $M_k \rightarrow \infty$, where $\Sigma_k$ is an asymptotic covariance matrix.
\end{proposition}

See \citeauthor{Caf05} for hypotheses and a proof sketch. Provided that we are able to estimate $\Sigma_k$, Proposition \ref{thm2:Caf_normality} allows us to build asymptotic confidence intervals for the EM increment, $\Delta Q$. Recall that in Section \ref{sec:GEM}, we defined the Generalized EM algorithm by requiring that $\Delta Q \geq 0$, and showed that this requirement guarantees the ascent property. While the stochastic nature of the MCEM algorithm makes it impossible to guarantee that the EM increment is positive, we are able to use Proposition \ref{thm2:Caf_normality} to construct asymptotic confidence bounds for $\Delta Q$. If a lower confidence bound for $\Delta Q$ is positive, then we can be reasonably confident that the true value of $\Delta Q > 0$.

Estimating the asymptotic variance under direct or rejection sampling is fairly straightforward. Importance sampling however, is somewhat more complicated; particularly when a normalizing constant must be estimated (see Section \ref{sec:imp_samp} for more on importance sampling). \citeauthor{Caf05} give a Delta Method-based formula for estimating $\Sigma_k$ under importance sampling \citep[see Chapter 3 of][for an overview of the Delta Method]{van98}. They also give some guidance for calculating standard errors based on Markov chain Monte Carlo sampling, which we do not go into here. See Section \ref{sec:MCMC} for more details on Markov chain Monte Carlo sampling. 

We now return to the key MCEM problems of choosing the Monte Carlo size and when to terminate. For the former, \citeauthor{Caf05} advise constructing a lower confidence limit for the EM increment, $\Delta Q$. If this confidence limit is positive, then we proceed to the next iteration. If not, then we augment the Monte Carlo sample at the current iteration (with, say, $M_k/r$ new points, for some small positive integer, $r$, as in \citealp{Boo99}), and compute a new confidence bound. At the next iteration, \citeauthor{Caf05} advise using a starting Monte Carlo sample which is at least as large as the final sample from the previous iteration. In fact, we may find that a larger sample should be used based on extrapolation of Monte Carlo variability from the previous iteration. The paper gives a formula to check for whether we should increase the Monte Carlo size before starting the next iteration based on a normal approximation to increments in the MCEM objective function. In our sample analyses, increasing the Monte Carlo size between iterations is never called-for, so we omit this step from our presentation.

An important difference between the methods of \citet{Boo99} and \citeauthor{Caf05} for determining Monte Carlo size, $M$, is that the former authors immediately proceed to the next iteration after increasing $M$, whereas the latter authors continue increasing $M$ at the current iteration until its size is deemed acceptable.

\citeauthor{Caf05} base their termination criterion on stopping when there is evidence that the algorithm is no longer yielding sufficient improvement in the EM objective function. Specifically, they start by choosing a tolerance, $\tau>0$, then calculate an upper confidence limit for the EM increment at each iteration. If this upper confidence limit is below $\tau$, then we declare that there is little remaining room for improvement in the EM objective, and terminate our algorithm.

We now apply the method of \citeauthor{Caf05} to our blood type example. This method has numerous tuning parameters, and the paper gives limited guidance on how to select them. As such, we choose values which appear to work reasonably well. Specifically, we use confidence levels of $80\%$ when checking whether to augment the Monte Carlo size and $90\%$ when checking for termination. Every time we augment the Monte Carlo size at iteration $k$, we add $M_k/2$ more points. We use a tolerance level of $10^{-3}$ to check for termination, and a starting Monte Carlo size of 10. Figure \ref{fig:blood_AMCEM_traj} gives trajectories of the MCEM estimates, as well as the Monte Carlo size used to obtain each of these estimates. As with the method of \citet{Boo99}, the trajectory stabilizes as the Monte Carlo size increases. Our final estimate here is $\hat{p} = 0.299$ and $\hat{q} = 0.127$, which is very close to the MLE. As with our other methods, the Monte Carlo fluctuations here are much smaller than the statistical uncertainty given in equation (\ref{eq:obs_lik_SE}).

\begin{figure}
    \centering
    \caption{Trajectory of estimates for $p$ and $q$, as well as Monte Carlo sample sizes, from the method of \citeauthor{Caf05}. Horizontal dashed lines give the maximum likelihood estimates.}
    \label{fig:blood_AMCEM_traj}
    \includegraphics[width=0.75\textwidth]{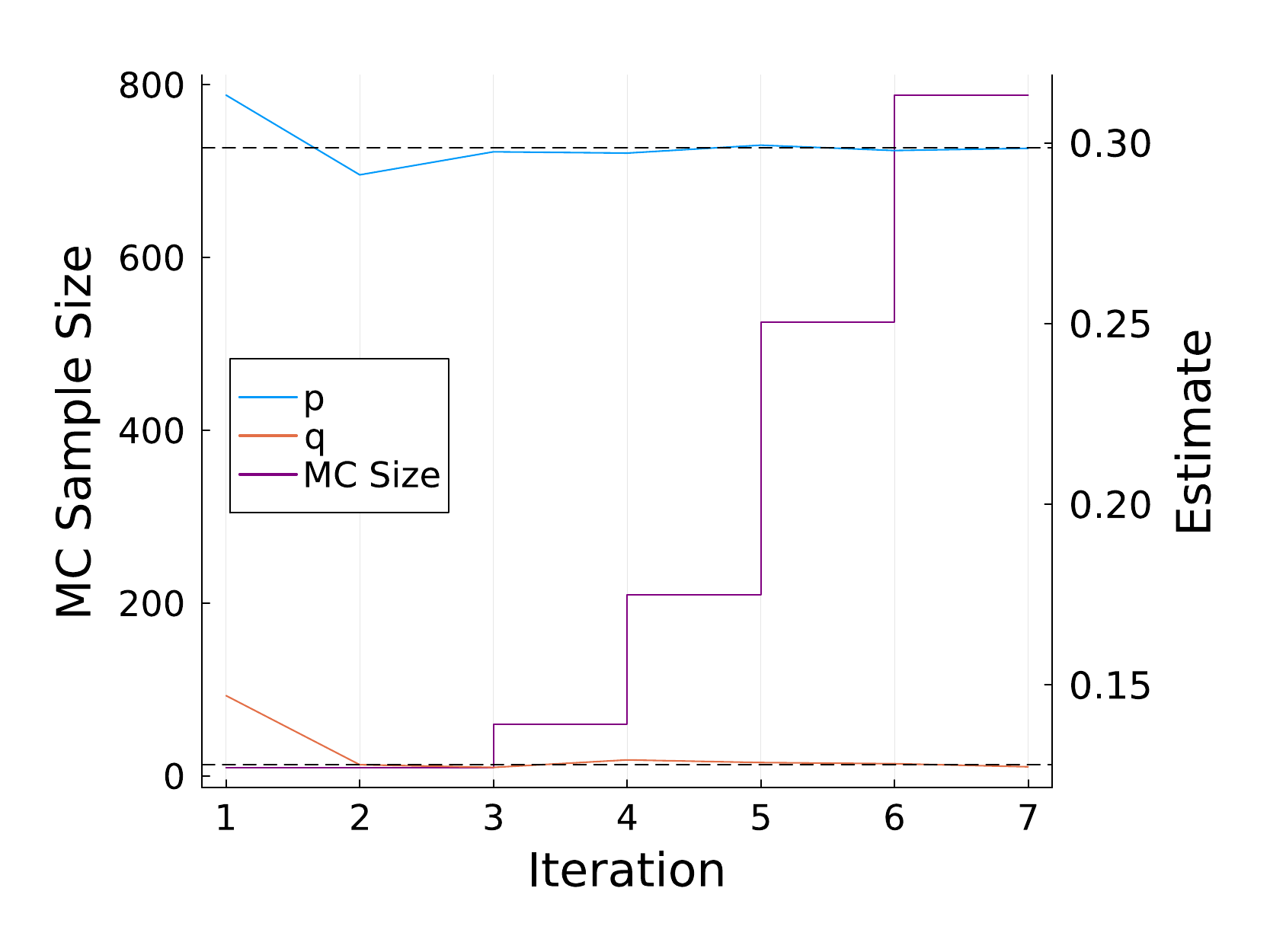}   
\end{figure}

\section{Alternatives to the MCEM Algorithm}
\label{sec:alternatives}

In this section, we outline some alternatives to the MCEM algorithm for maximizing the likelihood of an incomplete dataset. Examples include stochastic approximation \citep{Del99, Rob51, Lai03} and the Monte Carlo maximum likelihood method \citep{Gel93, Gey94}. 

There are also a few methods which are related to EM and MCEM, but have received little attention in the literature. One such method is the Monte Carlo Newton-Raphson algorithm \citep{Kuk97}, which closely resembles the ordinary Newton-Raphson algorithm, but with the gradient and Hessian replaced by Monte Carlo approximations. See our discussion of \citet{McC97} in Section \ref{sec:comparison} for a comparison with other related algorithms. Another related method is variational inference. While related to the EM algorithm, variational inference is sufficiently distinct that reviewing it is outside the scope of this work. See, e.g., \citet{Ble17} or \citet{Tsi08} for relevant review papers.

\subsection{Stochastic Approximation}
\label{sec:SAEM}

Stochastic approximation (SA) is a method originally proposed by \citet{Rob51} for finding roots of functions which can only be evaluated with noise. This method was expanded upon rapidly by, e.g., \citet{Kie52} into a method for derivative-free optimization, and by \citet{Dvo56} with a systematic theoretical framework. Since the mid-20th century, SA methods have grown into a thriving research area. See, e.g., \citet{Kus97} or \citet{Bor22} for textbook-length treatments, and \citet{Lai03} for a survey paper which focuses on applications to statistics.

The basic version of SA iteratively updates our estimate of the root, $\theta_*$, of some unobservable function, $\phi$, based on the value of a noisy realization of that function at the current estimate. Specifically, if $\theta_k$ is our estimate at iteration $k$ and $\hat{\phi} \approx \phi$, then our estimate at iteration $k+1$ is $\theta_{k+1} = \theta_k + \alpha_k \hat{\phi}(\theta_k)$, where $\alpha_k$ is a sequence which goes to zero at a particular rate. Since our sequence of weights goes to zero with $k$, the update terms become negligible in the limit and our estimate of $\theta_*$ stabilizes. The precise requirement for these weights is that $\sum_{k=1}^\infty \alpha_k = \infty$ and $\sum_{k=1}^\infty \alpha_k^2 < \infty$. A common choice is $\alpha_k = k^{-1}$. Numerous authors have studied convergence of the SA method in probability, almost surely and in $\mathcal{L}^1$ under various regularity conditions. See \citet{Lai03} for an excellent review of the history of stochastic approximation convergence theory.

Although the specific implementations of stochastic approximation are too numerous to cover here, we do give a particular application which is of great use to us. Suppose that we want to optimize a function, $f$, which we can only evaluate approximately. Assume further that we are able to approximately evaluate the gradient of $f$, $\nabla f$. Setting $\phi = \nabla f$ and running the SA algorithm with $\hat{\phi} \approx \nabla f$ gives us an approximate critical point for $f$. We can therefore use SA to optimize functions which cannot be evaluated exactly. This version of SA sees considerable use in the Machine Learning community under the name Stochastic Gradient Descent \citep{Bot10}. A related method developed by \citet{Kie52} involves approximating the gradient at each step with finite differences.

We now discuss specific applications of stochastic approximation to the missing data problem. We refer to such an application as a stochastic approximation EM (SAEM) algorithm. We turn first to the method of \citet{Gu98I}, which follows the outline presented above for using stochastic approximation-based optimization. That is, \citeauthor{Gu98I} suggest setting $\phi$ to the observed data score, $\phi(\theta) = S(\theta; y)$. This function can be estimated using Proposition \ref{thm2:EM_decomp} (i) and approximating the conditional expectation by Monte Carlo as in MCEM. Iteratively applying the stochastic approximation update formula converges to a critical point of the observed data score. In the case of a vector-valued parameter, \citeauthor{Gu98I} also recommend pre-multiplying $\hat{\phi}$ by a matrix which converges to the inverse of the observed data information matrix. Such a sequence of matrices can be constructed using Louis' Identity (Proposition \ref{thm2:info_decomp}) and Monte Carlo. Note that we use the same sample of missing data to update our estimates of the parameter and the observed data information matrix at each iteration. Similar work by \citet{Gu98II} extends the same SAEM construction to accommodate Markov chain Monte Carlo sampling \citep[see also][]{Cai10}. \citet{Gu01} discuss how to apply the above methodology to the analysis of spatial models, and incorporate a second stage to the method in which estimates are averaged across iterations (as recommended by \citealp{For03}, for MCEM). This averaging has been found to improve performance of SAEM, and of stochastic approximation more generally \citep{Pol92, Del99}, although \citet{Kuh05} report more modest findings.

\citet{Del99} present an SAEM implementation for estimation in exponential family models, in which stochastic approximation is used to estimate the EM objective function instead of working directly with $\theta$. Here, the estimate being updated at each iteration is $\tilde{Q}$, an approximation to the EM objective function, $Q$. The update term is $\hat{\phi}_k = \tilde{Q}_k - \hat{Q}_k$, where $\hat{Q}_k$ is the MCEM objective function based on the estimated value of $\theta$ from the previous iteration. Our updated parameter estimate is then obtained by maximizing over $\theta$ in the new stochastic approximation objective function, $\tilde{Q}_{k+1}(\theta) = \tilde{Q}_k(\theta) + \alpha_k [ \tilde{Q}_k(\theta) - \hat{Q}_k(\theta)]$. Note that the updated objective function can be re-written as a convex combination: $\tilde{Q}_{k+1} = (1 - \alpha_k) \tilde{Q}_k + \alpha_k \hat{Q}_k$. Since $\alpha_k \rightarrow 0$, each iteration of stochastic approximation progressively puts more weight on the pre-existing objective function and less weight on the MCEM objective.

Keen observers will note that the update formula given above does not fit exactly into the stochastic approximation framework given earlier in this section. Specifically, the formula given by \citet{Del99} does not directly update $\theta$, but instead updates $\tilde{Q}$, which depends indirectly on $\theta$, and is then used to infer an update for $\theta$. In order to re-frame the algorithm of \citeauthor{Del99} as a stochastic approximation update, we must use an approximate sufficient statistic as our estimate of $\theta$ at each iteration. We also add an asymptotically negligible bias term to our update formula (as a theoretical device). See Section 5 of \citet{Del99} and Chapter 5 of \citet{Kus03} for details.


A subtly different line of research on the SAEM method has been developed by a group at the National Institute for Research in Digital Science and Technology (INRIA) in France \citep[see, e.g.,][for a review of some of their methods]{Cel95}. The goal here is to augment the EM algorithm, rather than to facilitate the application of EM-type methods when ordinary EM is intractable. More precisely, methods from this group introduce a stochastic perturbation to the EM algorithm, with the goal of escaping fixed points which are locally, but not globally, optimal. Early work centered around a Stochastic EM (SEM) algorithm \citep{Cel85}, which is equivalent to the MCEM algorithm with a Monte Carlo size of one (\citealp{Cel87,Cel95}; see also \citealp{Nie00II}). Later, they also propose a method which they refer to as SAEM (although it does not quite fit into our framework), in which each iteration consists of first computing both the EM and MCEM updates from the previous iteration's estimate, then combining these two updates in a convex combination as the estimate for the current iteration. Here, as with the SEM algorithm, the MCEM update is computed with a Monte Carlo size of one \citep{Cel92, Cel95}.

An advantage of the SAEM algorithm over MCEM is that in SAEM we choose the Monte Carlo size once at the beginning and leave it fixed for every iteration. We can think of the method as automatically increasing the MC size since the estimate at each iteration is a weighted sum of all the estimates which came before it. A disadvantage of SAEM is that it requires us to select the sequence $\{ \alpha_k \}$, commonly referred to as the ``step size''. Choosing $\alpha_k$ too large will mean the algorithm takes a long time to stabilize, while choosing $\alpha_k$ too small causes the algorithm to stabilize before it reaches its limiting value (and will therefore take a long time to reach this limit). \citet{Jan06} gives some guidance on choosing this step size based on the goal of balancing bias with variance. \citeauthor{Jan06} also presents a convergence diagnostic based on the ideas of \citet{Caf05} which allows for a posteriori assessment of whether the step size was too small. 

We illustrate the SAEM methods of \citet{Gu98I} and \citet{Del99} on our blood type example. For both methods, we use $\alpha_k = k^{-0.7}$ (balancing step length with Monte Carlo variability), a Monte Carlo sample size of 10 at each iteration, and 50 iterations. Since our parameters are constrained to lie between 0 and 1, we apply the method of \citeauthor{Gu98I} on logit-scale, then back-transform before plotting\footnote{This logit transformation is not necessary for the method of \citet{Del99}, since there is always a maximizer of the estimated objective function which satisfies our parameter constraints.}. See Figure \ref{fig:blood_SAEM_traj} for trajectories from both SAEM methods. The final estimate from the \citeauthor{Gu98I} method is $\hat{p} = 0.291$ and $\hat{q} = 0.127$, while the \citeauthor{Del99} method gives $\hat{p} = 0.301$ and $\hat{q} = 0.128$. 

\begin{figure}
    \centering
    \caption{Trajectory of estimates for $p$ and $q$ from two versions of stochastic approximation. Horizontal dashed lines give the maximum likelihood estimates.}
    \subfloat[\citet{Gu98I}]{\includegraphics[width = 0.75\textwidth]{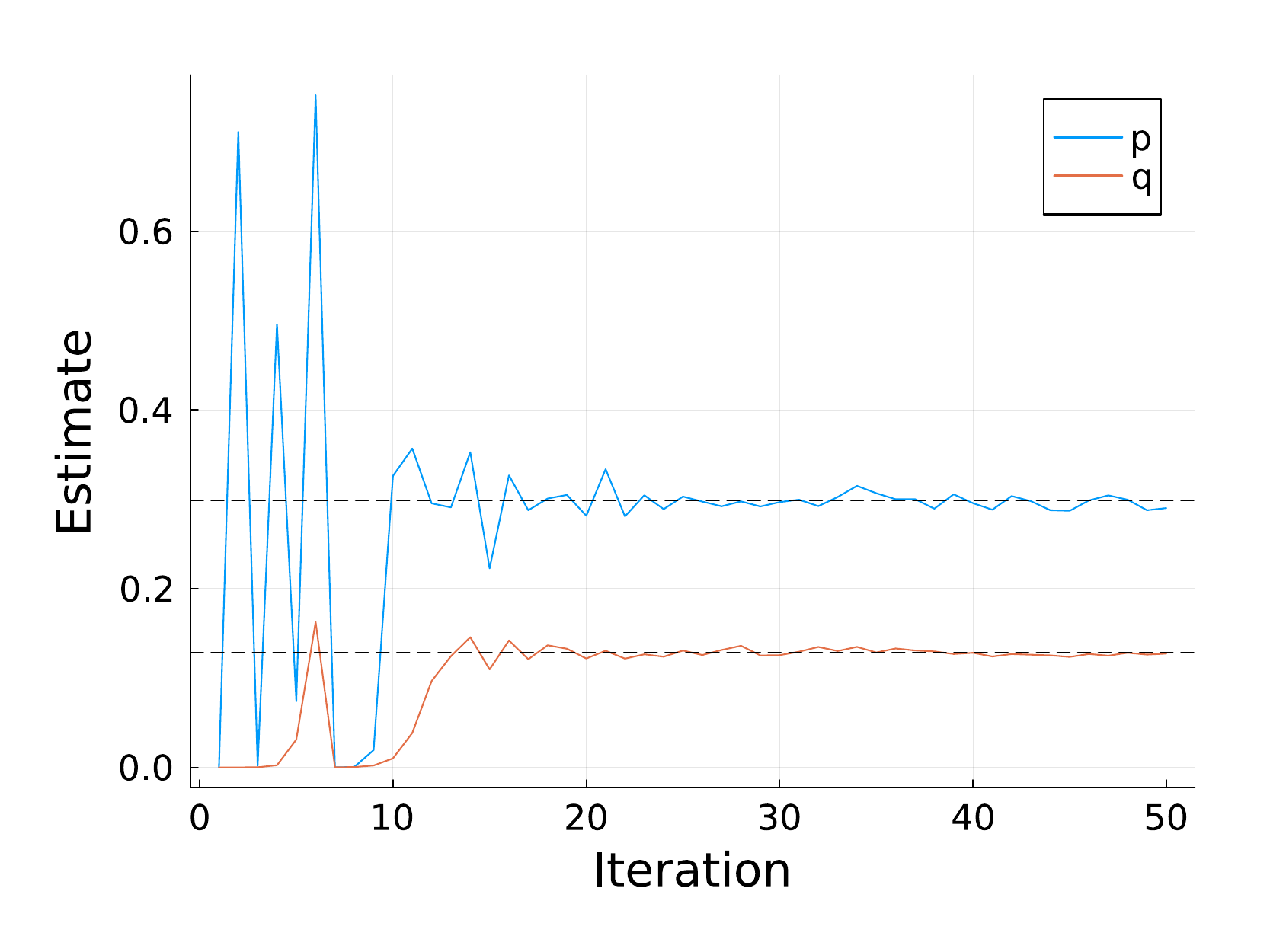} \label{fig:blood_SAEM_traj_score}}

    \subfloat[\citet{Del99}]{\includegraphics[width = 0.75\textwidth]{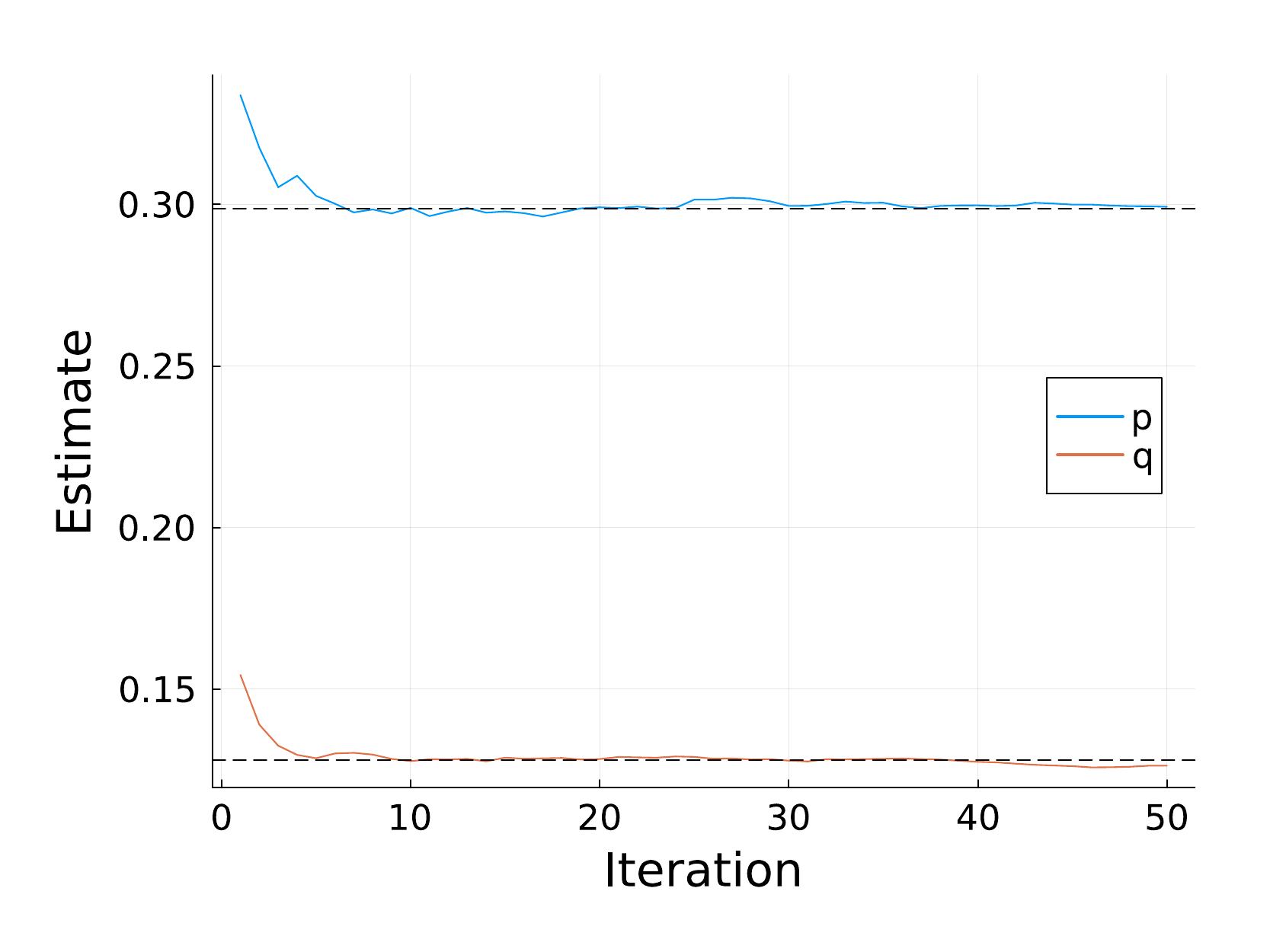} \label{fig:blood_SAEM_traj_obj_fun}}   
    
    \label{fig:blood_SAEM_traj}
\end{figure}

Both SAEM estimates are close to the MLE, but differ more than the MCEM methods presented in Section \ref{sec:MCEM}. It is worth noting that the number of Monte Carlo samples used to compute the final estimate under SAEM is comparable to that used by many MCEM methods. For SAEM however, this is the total number of Monte Carlo draws, whereas MCEM requires this many simulations at every iteration. Therefore, the total number of Monte Carlo draws used by SAEM is much lower than for MCEM.

\subsection{Monte Carlo Maximum Likelihood}
\label{sec:MCML}

The Monte Carlo maximum likelihood (MCML) method was developed to handle settings where the likelihood function cannot be evaluated exactly. The idea is to approximate the whole likelihood surface (up to an additive constant) using a single Monte Carlo sample. This approximate likelihood is then maximized, either numerically or analytically. We begin MCML by choosing a fixed reference parameter value, $\theta_*$, and estimate likelihood ratios relative to $\theta_*$. The likelihood ratio can be written as an expectation with respect to the fixed reference value, and this expectation is approximated by Monte Carlo. A key feature of this methodology is that a single Monte Carlo sample can be re-used to estimate the likelihood ratio at any number of target parameter values. Finally, we maximize our estimated likelihood ratio as a proxy of the unknown likelihood function.

The most basic form of Monte Carlo maximum likelihood \citep{Gey91} applies when we only know the likelihood up to a normalizing constant, say $f(y; \theta) = h(y; \theta) / c(\theta)$, where $h$ is known but $c$ is not. We note that $\int h(y; \theta) dy = c(\theta)$, and write
\begin{align}
    \log \frac{f(y; \theta)}{f(y; \theta_*)} &= \log \frac{h(y; \theta)}{h(y; \theta_*)} - \log \frac{c(\theta)}{c(\theta_*)}\\
    &= \log \frac{h(y; \theta)}{h(y; \theta_*)} - \log \int \frac{h(y; \theta)}{c(\theta_*)} dy\\
    &= \log \frac{h(y; \theta)}{h(y; \theta_*)} - \log \int \frac{h(y; \theta)}{h(y; \theta_*)} f(y; \theta_*) dy\\
    &= \log \frac{h(y; \theta)}{h(y; \theta_*)} - \log \bE_{\theta_*} \frac{h(y; \theta)}{h(y; \theta_*)}\\
    & \approx \log \frac{h(y; \theta)}{h(y; \theta_*)} - \log \frac{1}{M} \sum_{i=1}^M \frac{h(y_i; \theta)}{h(y_i; \theta_*)} \label{eq:MCML_simp}
\end{align}
where the $y_i$ are sampled iid from $f(y; \theta_*)$. The second term in line (\ref{eq:MCML_simp}) may need to be modified if non-iid sampling is used.

An alternative formulation of the MCML method given by \citet{Gel93} is more relevant for use with missing data \citep[see also][]{Gey94}. A similar derivation to the one given above shows that, in the presence of missing data, $X$,
\begin{align}
    \log \frac{f(y; \theta)}{f(y; \theta_*)} &= \log \bE_{\theta_*} \left[ \left. \frac{h(y, X; \theta)}{h(y, X; \theta_*)} \right| Y=y \right] - \log \bE_{\theta_*} \left[\frac{h(Y, X; \theta)}{h(Y, X; \theta_*)}\right] \label{eq:MCML_miss_exact}\\
    & \approx \log \left[\frac{1}{M_1} \sum_{i=1}^{M_1} \frac{h(y, X'_i; \theta)}{h(y, X'_i; \theta_*)}\right] - \log \left[\frac{1}{M_2} \sum_{j = 1}^{M_2} \frac{h(Y_j, X_j; \theta)}{h(Y_j, X_j; \theta_*)}\right] \label{eq:MCML_miss}
\end{align}
where the $X'_i$ are generated from the conditional distribution of $X|Y=y$, the $(Y_j, X_j)$ pairs are generated from the joint distribution of $Y$ and $X$, and $h(y, x; \theta)$ is proportional to this joint distribution. Note that two Monte Carlo samples are required to evaluate (\ref{eq:MCML_miss}). If the complete data likelihood is known exactly, then we can replace $h$ with $f$ in (\ref{eq:MCML_miss_exact}) and (\ref{eq:MCML_miss}), and drop the second term (i.e.\ the one being subtracted).

The MCML procedure closely resembles importance sampling, with $f$, $f_c$ or $f_m$ evaluated at the fixed $\theta_*$ being used as proposal for sampling from the same distribution with arbitrary $\theta$ (see Section \ref{sec:imp_samp} for a discussion of importance sampling). \citet{Jan03} take this idea further, and suggest directly estimating $f(y; \theta)$ from the complete data density, $f_c(y, x; \theta)$, using importance sampling. That is, they write $f(y; \theta) = \int f_c(y, x; \theta) dx = \int [ f_c(y, x; \theta) / g(x)] g(x) dx$, where $g$ is an arbitrary density function. The final integral is then approximated by an average of samples drawn iid from $g$. While the method of \citeauthor{Jan03} does present an interesting direction in which to generalize MCML, note that it is only applicable if $f_c$ is known exactly (or, rather, up to a proportionality constant which does not depend on $\theta$).

We apply the MCML method of \citet{Gey94} on our blood type example. Using a Monte Carlo size of 1000 gives $\hat{p} = 0.298$ and $\hat{q} = 0.129$, which are both quite close to the MLE. Comparisons with other methods have found that MCML is very sensitive to its starting point; see Section \ref{sec:comparison} or \citet{McC97}. We therefore repeat MCML, this time starting with the estimate obtained from its first application. Our second run gives $\hat{p} = 0.297$ and $\hat{q} = 0.128$, which is not appreciably closer to the MLE.

\section{Simulation}
\label{sec:simulation}

An obstacle to implementing the MCEM algorithm which was not addressed in Section \ref{sec:MCEM}, and which is also relevant to the methods described in Section \ref{sec:alternatives}, is how to generate the necessary Monte Carlo samples. It is in general a hard problem to simulate from arbitrary conditional distributions. In this section, we discuss a few methods for simulating the required observations at each step of the MCEM algorithm. All the topics that we cover here have their own bodies of literature, which we cannot hope to cover in their entirety. We instead give only a brief overview, focusing on aspects which are particularly relevant to use with the MCEM algorithm, and direct the reader to other, more focused, reviews.

We start by discussing importance sampling, which is a method for using a sample from one distribution to estimate moments of another distribution. Next, we cover Markov chain Monte Carlo (MCMC) sampling, which consists of constructing and sampling from a Markov chain whose stationary distribution matches the distribution from which we wish to simulate. Finally, we briefly touch on rejection sampling, sequential Monte Carlo (SMC) and quasi-Monte Carlo (QMC). We focus less on these last three methods because, to our knowledge, they have not been as widely used in the MCEM literature.

\citet{Del99} give some brief guidance on selecting between importance sampling and MCMC. \citet{Jan06II} goes into more detail about the advantages and disadvantages of these samplers in the context of MCEM. See also \citet{Rob04} for a textbook-length treatment of Monte Carlo methods.

\subsection{Importance Sampling}
\label{sec:imp_samp}

Broadly speaking, importance sampling is a framework for approximating intractable expectations. A typical use case is when we want to evaluate the expected value of some function, $h$, under a distribution $\bF$, and this expectation is not just analytically intractable, but the distribution $\bF$ is impossible (or impractical) to sample from. The latter restriction prevents us from using ordinary Monte Carlo integration. Instead, we select another distribution, $\bG$, which is easier to work with, and observe that $\bF h = \bG [h \cdot (f/g)]$, where $f$ and $g$ are the densities of $\bF$ and $\bG$ respectively. Provided that $\bG$ is easy to sample from, we can estimate this alternative expression for our target expectation via Monte Carlo integration with samples drawn from $\bG$. We call $\bF$ the target distribution and $\bG$ the proposal distribution.  

A classic reference on importance sampling and other Monte Carlo methods is the book by \citet{Rob04}; particularly Chapters 3 and 4. Chapter 8 of the book by \citet{Cho20} gives a more current overview of importance sampling, with a focus on its application to SMC. \citet{Elv22} give a survey of modern methods for extending the importance sampling framework, with an emphasis on two main approaches: multiple importance sampling and adaptive importance sampling. See \citet{Elv19} for a review of multiple importance sampling methods, and \citet{Bug17} for more on adaptive importance sampling. \citet{Aga17} give a survey paper level treatment of some more theoretical aspects of importance sampling.

In the rest of this section, we describe a few simple modifications which can ease implementation and improve performance when using importance sampling with the MCEM algorithm.

When using importance sampling with the MCEM algorithm, our target distribution is the missing data distribution (i.e.\ the conditional distribution of the missing data given the observed data). In some settings, this distribution may be difficult to describe exactly. However, the integrand is the (log-)likelihood of the complete data distribution, so it is reasonable to expect that we can evaluate this complete data density. From the definition of conditional probability, the missing data density is proportional to the complete data density, provided that we treat the latter as a function of the missing data and hold the observed data fixed. The proportionality constant here is the observed data density, so it is unlikely that we will be able to normalize the missing data density exactly (otherwise, we could just work directly with the observed data likelihood). 

A simple modification of importance sampling, which circumvents the need for exact normalization, is to compute importance weights using un-normalized densities, then normalize them to sum to one. This is referred to as ``auto-normalized'', or ``self-normalized'', importance sampling \citep[see, e.g.,][]{Elv22}. The reason self-normalized importance sampling works is that the unknown normalizing constant cancels in the numerator and denominator of our normalized weights. In fact, our proposal distribution can also be un-normalized, and the ratio of the two normalizing constants cancels when we normalize our weights.

There are, however, disadvantages of self-normalized importance sampling compared to the exact importance sampling. One important limitation is that our estimator of $\bF h$ is no longer unbiased, and is instead only asymptotically unbiased. In fact, as a ratio estimator, the standard error of a self-normalized importance sampling estimator is often difficult to obtain for finite samples (although an asymptotic formula is readily available). This is fine in problems where it is easy to sample from our proposal distribution, $\bG$, but in high dimensional problems for example, even simulating from $\bG$ may be costly, and more care must be taken with the discrepancy between asymptotic results and finite-sample behaviour.

It is well-known that the performance of an importance sampling estimator depends on how closely the proposal distribution matches the target \citep{Rob04}. One thing that can go wrong is if the importance weight (i.e.\ the likelihood ratio between these two distributions) does not have sufficiently many finite moments \citep{Aga17}. A simple modification to our importance weights which guarantees infinitely many finite moments is to specify a threshold value, and truncate any weights which fall above the threshold (i.e., set weights which fall above this threshold equal to the threshold value). This ``truncated importance sampling'' method is proposed and analysed by \citet{Ion08}. Two strategies are proposed in this work for selecting the threshold: the first is to simply use the square root of the Monte Carlo size, $\sqrt{M}$, while the latter is based on unbiased risk estimation and gives a value better tailored to the specific problem. Note that these recommendations are based on exact importance sampling. If self-normalization is used, the general recommendation is instead to truncate at $\sqrt{M}$ times the mean of the un-normalized weights. Choosing the threshold level for truncated importance sampling requires managing the bias-variance trade-off \citep{Has09}. Truncating weights reduces the variance of our importance sampling estimator, but also introduces bias. This trade-off highlights the importance of selecting an appropriate threshold value: too large and the variance reduction will be negligible, but too small and the bias will be unacceptable.

\citet{Veh22} propose an alternative method for handling large importance weights, called Pareto Smoothed Importance Sampling (PSIS). The idea of this method is to fit a Generalized Pareto Distribution to the largest importance weights, then replace those large weights with quantiles of the fitted distribution. One of the parameter's fitted values also serves as a useful diagnostic for how closely our proposal distribution matches the target. An advantage of PSIS over the truncated importance sampling method of \citet{Ion08} is that the bias of PSIS is smaller, although it is not clear in general which method has better mean-squared error \citep[see][for extensive numerical comparisons]{Veh22}.

An approach presented by \citet{Qui99} and \citet{Lev01} seeks to use importance sampling to save computing time when running the MCEM algorithm. This computational efficiency is especially important in their context where Markov chain Monte Carlo sampling is used (see \ref{sec:MCMC}), as generating a single sample using this method often takes quite some time. Their idea is to generate a single sample from the missing data distribution with some reference value for $\theta$, $\theta_{ref}$. This sample is then re-used at every MCEM iteration, with conditional expectations under the current parameter value computed by taking importance ratios with respect to $\theta_{ref}$. An important question here is whether a single proposal distribution can be adequate for every target distribution along the MCEM trajectory. To address this concern, both sets of authors suggest running MCEM for a few iterations with fresh importance samples and only drawing the sample which will be used for their method after a sufficiently long ``burn-in'' period (\citeauthor{Qui99} use five MCEM iterations, while \citeauthor{Lev01} run the algorithm for one minute).

Other applications of importance sampling in the literature on MCEM and related methods include \citet{Caf05} incorporating importance weights into the formulas of their estimators. Recall also that the essential idea of Monte Carlo maximum likelihood is to use importance sampling to estimate the observed data likelihood (see Section \ref{sec:MCML}).

\subsection{Markov Chain Monte Carlo}
\label{sec:MCMC}

The core idea of Markov chain Monte Carlo (MCMC) sampling is to construct a Markov chain from which we can simulate, and which has stationary distribution equal to the target distribution. Popular methods to construct such a Markov chain are the Metropolis-Hastings algorithm and Gibbs sampling. See \citet{Gel13} or \citet{Rob04} for excellent textbook-length overviews. A popular implementation of MCMC sampling is the \texttt{Stan} programming language \citep{Sta22}, and its \texttt{R} interface, \texttt{RStan} \citep{Sta23}.

Both \mh\ and Gibbs sampling start with a random variable, $X = (X_1, \ldots, X_d)$, which we wish to simulate. Let $f$ be the density of $X$. These methods proceed by iteratively simulating draws of the vector $X$ from a Markov chain, whose distribution depends on the draw from the previous iteration. Note that our sampler may need time to reach its stationary distribution. It is therefore common to use a ``burn-in'' period, which amounts to discarding some number of draws from the beginning of the sample.

The Metropolis-Hastings algorithm closely resembles rejection sampling (see Section \ref{sec:other_samplers}). We start each iteration by generating a candidate value of $X$ from some proposal distribution (this distribution may depend on the previous iteration's value of $X$). Write $J(x | x_0)$ for the proposal density, where $x_0$ is the value of $X$ from the previous iteration. Next, we define an acceptance probability, $r := f(x) J(x_0 | x) / f(x_0) J(x | x_0)$, and accept the proposed value of $X$ with probability $r \wedge 1$. With probability $(1-r) \vee 0$, we reject the proposed $x$ and instead set the current iteration's value to $x_0$. We then proceed to the next iteration. Note that rejecting still adds an observation to our Monte Carlo sample, this value just happens to be identical to the one proceeding it. A popular extension of the \mh\ algorithm is Hamiltonian Monte Carlo \citep[see Section 12.4 of][]{Gel13}, in which the Markov chain is updated based not only on the current proposal, but also using a ``momentum'' term which is updated concurrently at each iteration.

Gibbs sampling is based on successively sampling each component of $X$ conditional on all the other components. The order in which this conditioning is performed is a bit subtle however. When generating $X_i$, we condition on the values of $X_1,\ldots, X_{i-1}$ from the current iteration and the values of $X_{i+1}, \ldots, X_d$ from the previous iteration. Once we reach the end of $X$, we start a new iteration.  More generally, we can group the components of $X$, and simulate an entire group conditional on the others. In practice, it may not always be easy to simulate directly from the necessary conditional distributions. One solution to this problem is the Hybrid-Gibbs (or Metropolis-within-Gibbs) sampler \citet[see Section 10.3.3 of][]{Rob04}, in which one step of the \mh\ algorithm is used at each stage of Gibbs sampling to facilitate necessary simulation.

Limit theory for estimators based on MCMC sampling are more complicated than the similar theory for estimators based on iid or importance sampling \citep{Gey91}. This additional complexity stems from the dependence between draws from an MCMC sampler. While convergence for iid and importance sampling can be established using the Law of Large Numbers and the classical Central Limit Theorem, similar results for MCMC sampling make use of the Ergodic Theorem and Markov chain Central Limit Theorem. A challenge in implementing the Markov chain Central Limit Theorem is that the asymptotic variance depends on pairwise covariances between points in the chain with arbitrarily large lags. Estimation of this asymptotic variance is therefore challenging in practice. See, e.g., Chapters 6 and 7 of \citet{Rob04}.

\citet{Lev01} propose a method to simplify the analysis of estimates based on MCMC sampling. Their approach consists of subsampling the original chain at random lags (called Poisson spacings) in such a way that elements of the subsample are approximately independent. The result is that a target function averaged over the subsampled chain satisfies the classical Central Limit Theorem (i.e. with no covariance terms in the asymptotic variance). We can estimate the mean and standard error of our subsample estimator using the entire chain, then construct confidence intervals for the mean of the target function based on our subsample. \citeauthor{Lev01} apply this strategy for constructing confidence intervals similarly to the method used by \citet{Boo99}. Here, we construct a confidence interval for the gradient of the EM objective function, and increase the Monte Carlo size if the confidence interval for the current iteration contains the subsample estimate from the previous iteration. This work is modified to provide uncertainty quantification for the MCEM update instead of the EM score by \citet{Lev04}. Furthermore, \citeauthor{Lev04} give a principled argument for how much to increase the Monte Carlo sample size when doing so is deemed necessary.

The convergence theory for SAEM is extended by \citet{Kuh04} to accommodate MCMC sampling. \citet{Tre14} and \citet{Bae16} then investigate the use of MCMC sampling with the MCEM algorithm and related methods in the context of fitting a complicated model for plant growth to real data on sugar beets. They find that MCMC is both accurate and computationally efficient in their particular application.

\subsection{Other Sampling Methods} \label{sec:other_samplers}

While importance sampling and MCMC are the most discussed sampling schemes in the context of the MCEM algorithm, some others do exist. Rejection sampling, sequential Monte Carlo (SMC) and quasi-Monte Carlo (QMC) are some such alternatives. However, these methods do not appear to be as frequently employed with the MCEM algorithm and related methods.

Rejection sampling closely resembles importance sampling, except instead of weighting each proposal by the likelihood ratio, we either accept or reject the proposed observation with probability proportional to the likelihood ratio \citep{Rob04}. Typically, rejection sampling is continued until the number of accepted proposals reaches a desired sample size. These accepted points are then treated as an iid sample from the target distribution. Although the output of rejection sampling sounds ideal (much of the difficulty with importance sampling comes from having to account for simulated points not having been drawn from the target distribution), the cost comes in increased computation time. The number of draws from the proposal required to get a fixed number of accepted draws is random, and can be quite high if the proposal distribution does not closely match the target. Indeed, if we instead fix the number of draws from the proposal distribution, importance sampling can be shown to have lower variance than the corresponding rejection sampling scheme \citep[see Section 8.8 of][]{Cho20}. \citet{Boo99} discuss the use of rejection sampling with their implementation of MCEM, but found that importance sampling was faster and gave similar results.

SMC, is a form of adaptive sampling in which a sequence of samples is generated such that the distribution of these samples converges to some target. The update from one sample to the next seeks to balance improving the proposal with maintaining computational efficiency. See \citet{Del06} for a survey paper, or \citet{Cho20} for a book-length overview of SMC. \citet{Tre13} uses several versions of SMC to generate Monte Carlo samples for the MCEM algorithm in their analysis of a model for plant growth. Follow-up work by \citet{Tre14} however, suggests that MCMC is more effective in this context. \citet{Mof14} use SMC in an MCEM analysis to sample from a truncated multivariate normal distribution with complicated support.

QMC is a modification of ordinary Monte Carlo, in which points are chosen deterministically instead of being sampled randomly \citep{Caf98}. The goal of this deterministic selection is to cover the sample space in a way that is somehow optimal. This optimality is achieved by using one of many ``low-discrepancy sequences''. An important limitation of QMC methods is that it is often challenging to measure their accuracy. This challenge is addressed by randomized QMC \citep{LEc02}, in which the starting point of the low-discrepancy sequence is chosen randomly, thereby making the estimated expectation random. Importantly, incorporating randomness into QMC allows us to measure the accuracy of our estimate by computing the variance across some number of independent runs. \citet{Jan05} illustrates the use of randomized QMC with the MCEM algorithm of \citet{Boo99}. He finds that randomized QMC is much more efficient than ordinary Monte Carlo on a spatial statistics problem with fixed computational effort, even after dividing this computational budget among multiple independent runs of randomized QMC to facilitate variance estimation.

\section{Comparisons Between Methods}
\label{sec:comparison}

Numerous authors have performed comparisons between methods discussed in Sections \ref{sec:MCEM} and \ref{sec:alternatives}. In this section, we discuss these comparisons and their findings.

\citet{McC97} uses a simulation study to compare the Monte Carlo EM (MCEM) and Monte Carlo Maximum Likelihood (MCML) methods, along with a Monte Carlo version of the Newton-Raphson algorithm (MCNR). Their MCEM implementation starts with fixed Monte Carlo size, then increases this size at iterations 20 and 40. It is not clear how these jump points were selected, nor how convergence was assessed beyond examining plots. Their MCEM implementation thus most closely resembles that of \citet{Wei90}. A similar schedule of Monte Carlo sizes and termination was used for MCNR, whereas MCML uses a much larger Monte Carlo size (the sample size for MCML was not increased). \citeauthor{McC97} also investigates the use of MCML after MCEM and MCNR (i.e. using MCEM or MCNR to choose the reference parameter for MCML). This comparison is made using a logit-normal mixed-effects model with one random effect and one fixed effect. \citeauthor{McC97} found that MCEM and MCNR perform better than MCML alone, but that following either MCEM or MCNR with MCML was even better. They did not find that following-up with multiple iterations of MCML was preferable to a single run of MCML.

In addition to presenting their stochastic approximation EM (SAEM) method, \citet{Gu98I} compare this method with MCEM. They give a comparative analysis on a dataset of motorette failure times \citep[see][for a diagram and explanation of what a motorette is]{Rai16}. The statistical model used is a linear regression with right-censoring. Their MCEM implementation is the same as that of \citet{Wei90} (more precisely, they say their implementation ``is from \citet{Tan93}'', which matches \citeauthor{Wei90}). It is not clear what Monte Carlo size they use to start, and it appears that this size is never augmented. Their SAEM implementation is as previously described in their paper (see Section \ref{sec:SAEM}), with an MC size of 1, step size at iteration $k$ (i.e., $\alpha_k$) of $1/k$, and pre-multiplying matrix chosen adaptively based on the current iteration, as given by Equation (13) of their paper. \citeauthor{Gu98I} find that SAEM converges much more quickly than MCEM. In fact, based on their Figure 1, it is not clear that MCEM is converging to the MLE at all. These authors also give a heuristic argument that SAEM should converge much more quickly than MCEM based on the number of MC samples used at each iteration and the number of iterations required to converge to the MLE. However, this argument is based on a fixed Monte Carlo size at each iteration and is thus not directly relevant to the MCEM implementations discussed in Section \ref{sec:MCEM}.

\citet{Boo99} compare their MCEM implementation with that of \citet{Wei90} as presented in \citet{McC97}. They investigate performance on three datasets: the logit-normal mixed-effects model from \citet{McC97}, a dataset comparing smoking with lung cancer \citep{Dor54} and the salamander dataset given in \citet{McC89}. On \citeauthor{McC97}'s dataset, \citeauthor{Boo99} ran their own method to convergence, then ran \citeauthor{Wei90}'s method for the same amount of time. They found that their method converges more quickly to the observed data MLE than the method of \citeauthor{Wei90} does (the model here is sufficiently simple that the observed data MLE can be obtained directly). They also investigate ``pure Monte Carlo error'' by starting both methods at the observed data MLE (since the MLE is a fixed point of EM, any change here is error due to Monte Carlo variability), and find that their method performs better on this metric as well.

Note that the comparisons made by \citeauthor{Boo99} use rejection sampling or importance sampling for their own method and Markov chain Monte Carlo (MCMC) sampling for the method of \citet{Wei90} (see Section \ref{sec:simulation} for a discussion of these simulation techniques). Of note is that MCMC typically requires a ``burn-in'' period to reach the required stationary distribution, so the results presented by \citeauthor{Boo99} may not show the best that we can expect from MCMC.

\citet{Boo01} compare the MCEM, SAEM and MCML (referred to as stochastic maximum likelihood, or SML) methods. They perform their comparison on a simple one-way mixed-effects linear model, with known variance component and error variance. Their MCEM implementation matches that of \citet{Boo99}, while their SAEM implementation is that of \citealp{Del99}. Note that although the presentation of SAEM in \citeauthor{Boo01} appears different from ours, it is not hard to show that the two are equivalent. Their MCML implementation uses the missing data distribution with a fixed value for the unknown parameter as reference distribution. The comparison between these methods is done partly analytically and partly by simulation. The mean squared errors (MSEs) of MCEM and SAEM for reproducing the MLE can be obtained analytically. \citeauthor{Boo01} thus compare these two methods directly and find that MCEM performs better when the problem is harder (i.e. larger variance component). The MSE of MCML on the other hand, must be approximated by simulation. This MSE is estimated using 500 replicates, and a $95\%$ Wald-type confidence interval is constructed. The upper and lower bounds of this confidence interval are then compared to the analytical MSE of MCEM (the authors do not compare MCML with SAEM). The result is that MCML is competitive with MCEM, and that MCML even performs better for some parameter settings (typically, when the variance component is small). \citeauthor{Boo01} point out that the comparison between MCEM and MCML is not entirely fair here though, because taking the variance component as known makes simulation for MCML unrealistically easy. In more serious problems, difficulty in choosing a proposal distribution for MCML will likely lead to worse performance.

\citet{Boo01} extend the heuristic argument in \citet{Gu98I} to account for Monte Carlo sizes changing with MCEM iteration. They argue that, in the scalar case, MCEM should converge more quickly than SAEM when the so-called ``fraction of missing information'' is larger than $\exp(-1)$. The fraction of missing information is defined as $\mathcal{I}_c(\hat{\theta})^{-1} \mathcal{I}_m(\hat{\theta})$, where $\hat{\theta}$ is the observed data MLE (see Proposition \ref{thm2:EM_decomp} for definitions of $\mathcal{I}_m$ and $\mathcal{I}_c$). This quantity is closely related to the convergence rate of the EM algorithm \citep{Men94,McL08}. \citeauthor{Boo01} also find that MCML outperforms both MCEM and SAEM. However, they use a proposal distribution for MCML which contains information about the parameters, thereby giving this method an unfair advantage. These findings are consistent with those of \citet{McC97} about MCML, where this method can perform very well, but is highly sensitive to the choice of proposal distribution.

\citet{Jan03} extend the work of \citet{Boo01}, specifically the comparison between MCEM and MCML. \citeauthor{Jan03} focus on analytical comparisons; as such, they use a fixed Monte Carlo size across iterations to make their calculations more tractable. They also use an unrealistic proposal distribution for MCML which requires that we know the observed data MLE. The authors derive the asymptotic variance of MCEM and MCML, and investigate the asymptotic relative efficiency (ARE) of these two methods. It turns out that the ARE depends directly on the eigenvalues of the (matrix-valued) fraction of missing information defined in the previous paragraph. In particular, the efficiency of MCEM relative to MCML goes to infinity as the fraction of missing information goes to 1. That is, as a problem gets harder in the sense that less information is available in the observed data, we expect MCEM to perform better relative to MCML. We also expect MCEM to perform still better on real problems, since the above analysis is based on an inaccessible proposal distribution for MCML. \citeauthor{Jan03} illustrate their analytical calculations on the one-way mixed-effects linear model from \citet{Boo01} and a logistic-normal generalized linear mixed-effects model. The latter consists of a simulation study which compares, among other things, the effect of the choice of proposal distribution on MCML. They found that the average estimates from MCEM and MCML (i.e. averaged over simulation replicates) are fairly consistent, but that variability of the MCML estimates is much higher. In particular, entries in the empirical covariance matrix for MCML grow rapidly as the reference parameter value for the proposal distribution moves away from the MLE.

\citet{Caf05} compare their version of MCEM with that of \citet{Boo99} on simulated data from the logit-normal mixed-effects model of \citet{McC97}. Specifically, \citeauthor{Caf05} simulate 10000 datasets and compare their method to the one proposed by \citet{Boo99}, with the latter terminating after the change in estimated parameter is small for between 1 and 4 consecutive iterations. They compare methods on how many draws from the missing data distribution are used (a measure of computational cost), the fraction of time spent in the final iteration, and the relative error in estimating both the MLEs of the parameters and their covariance matrix. \citeauthor{Caf05} find that, for a fixed amount of computing, their method performs a bit worse than \citet{Boo99} on estimating the parameters, but that their estimates of the covariance matrix are more accurate. They also find that their method spends a much larger fraction of its time in the final iteration. They argue that this is an advantage, since the Monte Carlo sample used at the final MCEM iteration can then be used to estimate moments of the missing data distribution.

\citet{Tre14} investigate the performance of MCEM on multiple axes for fitting a complex model of plant growth to date on sugar beets. First, they compare sequential Monte Carlo (SMC) sampling with MCMC (specifically, hybrid Gibbs sampling) using a fixed MC size at each iteration, as in \citet{Wei90}. From this comparison, they find that MCMC outperforms several popular versions of SMC. See \citet{Tre13} for more on their implementations of SMC. \citeauthor{Tre14} then implement the MCEM method of \citet{Caf05} and use a Monte Carlo study to compare the effect of various tuning parameters. Finally, they repeat their comparisons on a real data analysis, with 50 independent runs of each method for comparison. Average parameter estimates are very similar across all methods, but they do find differences in the variance over independent runs. Specifically, MCMC again outperforms SMC for MCEM with fixed Monte Carlo size, although at the cost of mildly increased computation time. They also find that the MCEM implementation of \citet{Caf05} gives similar performance as naive MCEM with much less computation time, or better performance with only somewhat less computation, depending on how tuning parameters are set.

\citet{Bae16} compare the SAEM method of \citet{Del99} with a version of MCEM similar to that given in \citet{McC97} on a model for the growth of sugar beet plants. \citeauthor{Bae16} use two different versions of MCMC sampling (Metropolis-Hastings and hybrid Gibbs), and two different proposals. Ultimately, they find little difference in the accuracy of their SAEM and MCEM implementations on either simulated or real data, but that SAEM requires much less computing time. They ultimately recommend SAEM, with the disclaimer that a more sophisticated implementation of MCEM \citep[e.g.,][]{Caf05,Boo99} might reduce the discrepancy between methods.

\subsection{Synthesis}
Many of the comparisons we have discussed use relatively naive implementations of MCEM \citep{McC97, Gu98I,Jan03}. In particular, as far as we can tell, \citet{Gu98I} use a poor implementation which never increases the Monte Carlo size, so it is unsurprising that they found MCEM performs poorly. \citet{Jan03} also use a fixed Monte Carlo size, but theirs is sufficiently large that we can expect the behaviour of MCEM to be close to that of the deterministic EM algorithm. Similarly, \citet{McC97} increase the Monte Carlo size at fixed iterations, and the final size is quite large (5000 for the final ten iterations).

Despite these concerns about Monte Carlo size, MCEM tends to perform quite well in simulations and analytical comparisons. Only \citet{Gu98I} and \citet{Boo01} found an instance where MCEM performed substantially worse than another method, and both of these simulation studies had features which biased results away from MCEM. Additionally, \citep{McC97} found that following MCEM with MCML tends to improve performance over MCEM alone. We see this as an endorsement of MCEM, since MCML is known to be sensitive to how well the proposal distribution (i.e. the output of MCEM) approximates the target distribution. \citet{Bae16} find that SAEM is more computationally efficient for their problem than MCEM, but their implementation of MCEM does not adapt the Monte Carlo size across iterations and is therefore wasteful of computational effort in early stages. These findings suggest that the MCEM algorithm should be one of the first methods considered when approaching a missing data problem in which the calculations required to implement EM are intractable.

When selecting which implementation of the MCEM algorithm to use, unfortunately, limited information is available. One example of such a comparison is given by \citet{Tre14}, in which they find that the MCEM implementation of \citet{Caf05} performs better than that of \citet{Wei90}. \citet{Caf05} offer another such comparison between their method and that of \citet{Boo99}, although the findings are not conclusive for one method over the other. Finally, \citet{Bae16} compare the MCEM method of \citet{McC89} with the SAEM method of \citet{Del99}, and find that SAEM is more computationally efficient, but neither method is noticeably more accurate. Given the breadth of implementations for the MCEM and SAEM algorithms, this is a very limited basis for choosing a method in practice. The limited range of comparisons between methods is a major gap in the literature on MCEM and related methods. We return to this point in Section \ref{sec:conc}.

\section{Conclusions}
\label{sec:conc}

The EM algorithm is a very useful tool for the analysis of missing data. The MCEM algorithm and related methods allow us to apply ideas from the EM algorithm in contexts where analytical calculations are intractable. In this paper, we present several implementations of the MCEM algorithm, as well as some alternative methods which can be applied to missing data problems. We also address the practical concern of how to generate the Monte Carlo samples required for our methods. Specifically, we discuss several methods for simulating from an arbitrary target distribution when direct sampling is not available. Finally, we discuss numerous comparisons between the MCEM algorithm and related methods. This gives both practitioners and researchers a basis for selecting which method to use in an analysis or for further study.

Over the course of writing this review, we have identified a number of gaps in the literature. Exploring these areas would be a significant contribution to the understanding and practice of the MCEM algorithm and related methods.

Firstly, there are many papers which include comparisons between MCEM and its alternatives. However, many of these comparisons are in the context of showing that a newly proposed method is effective. Papers which do focus on comparing existing methods tend to be limited in scope, including only a few methods and at most a few datasets. It would be valuable to have a more systematic comparison, in which many methods are compared (e.g., those discussed above) on a wide variety of datasets. One goal of such an analysis would be to identify features of a dataset which recommend the use of one method over another. A related limitation of the literature we have reviewed here is that most comparisons between models are empirical. Some work has been done to investigate theoretical error rates for a small range of algorithms \citep{Boo01, Jan03}, but more investigation of this form would further our understanding of systematic differences between methods. Furthermore, many methods we review include tuning parameters, which might govern the amount of Monte Carlo samples used or when to terminate the algorithm. Some effort has been made to help select values for these tuning parameters \citep[e.g.,][]{Jan06,Bae16}, but in general this is a difficult problem. Since the goal of tuning is typically to balance performance with computational effort, it can be computationally prohibitive to use the usual strategy of trying several options and seeing which looks best. Thus, it would be extremely valuable for practitioners to have general guidelines on how to set these tuning parameters.

Another valuable direction for future work is in the theoretical underpinnings of the MCEM and SAEM algorithms. While some such work does exist \citep[e.g.,][]{Del99, For03}, most justification for new methods is heuristic and empirical. Specific directions to explore include extending existing limit theory for MCEM and SAEM beyond (curved) exponential family models. This is especially important for the SAEM method of \citet{Del99}, since the link between their algorithm and the classical stochastic approximation theory depends on this exponential family structure. Another interesting direction is presented by \citet{Nie00}, in which iterations of the MCEM algorithm are viewed as consecutive observations of a Markov chain. This observation is leveraged to obtain an asymptotic covariance for MCEM estimators, both when estimating the MLE (i.e., conditional on the observed data) and the parameter of interest. This work is done specifically with fixed Monte Carlo size at each iteration (and growing with the observed data sample size when the latter goes to infinity), which is inconsistent with most practical MCEM algorithms.

As with any research project, there are always new directions to modify and extend the MCEM algorithm. For example, much attention has centered on the implementations of \citet{Boo99} or \citet{Caf05}, as well as that of \citet{Wei90} or \citet{McC97}. The former are likely popular due to their inferential flavour, while the latters' popularity is surely due to their simplicity. However, the method of \citet{Cha95} is quite different from any others described here and, although their method is mentioned often (at time of writing their paper has 435 citations on Google Scholar), it rarely shows up in empirical comparisons with other algorithms. This lack of attention is especially surprising since the variance estimator proposed by \citeauthor{Cha95} enjoys an unusually fast convergence rate (the estimated standard deviation is $m$-consistent instead of the more common $\sqrt{m}$-consistent, where $m$ is the Monte Carlo sample size). Another useful modification of the MCEM algorithm is to incorporate quasi-Monte Carlo sampling \citep[see][]{Jan05}. In principle, quasi-Monte Carlo sampling can be a low-effort way to dramatically reduce Monte Carlo variability. While quasi-Monte Carlo sampling has been mentioned elsewhere in the context of the MCEM algorithm \citet{Kuo08}, we are not aware of any papers focusing on this combination other than \citet{Jan05}.

We hope that this review helps make the MCEM algorithm, and Monte Carlo methods for missing data more generally, more accessible. We also hope that our work generates increased interest and development in the important field of computational methods for missing data.

\newpage

\begin{appendices}

    \section{Likelihood for Gene Frequency Estimation}
    \label{app:blood}

    In this appendix, we present details for the analysis of our example of estimating gene frequency. See Section \ref{sec:eg-genes} for formulation of the model and definition of notation.

    \subsection{Observed Data Likelihood, Score and Information}
    \label{app:blood_obs}

    Let $\pi_i$ be the probability of blood type $i$. The observed data log-likelihood for our model can be written as follows:
    \begin{align}
        \ell(\theta; y) &= \log \begin{pmatrix} n \\ y \end{pmatrix} + \sum y_i \log \pi_i(\theta)\\
        & \equiv \sum y_i \log \pi_i\\
        &\equiv 2 y_1 \log r + y_2 \log(p^2 + 2pr) + y_3 \log(q^2 + 2qr) + y_4 \log pq
    \end{align}
    where we use $\equiv$ to denote equality up to additive constants which do not depend on $\theta$.

    Differentiating $\ell$ with respect to $\theta$ and recalling that $r = 1 - p - q$, so $\partial_p r = \partial_q r = -1$, we get the following expression for the observed data score, $S$.
    \begin{align}
        S(\theta; y) &= \begin{pmatrix}
            \partial_p \ell(\theta; y)\\
            \partial_q \ell(\theta; y) 
        \end{pmatrix} \mathrm{, where}\\
        \partial_p \ell(\theta; y) &= - \frac{2 y_1}{r}  + \frac{2r y_2}{p^2 + 2pr}  - \frac{2q y_3}{q^2 + 2qr}  + \frac{y_4}{p} \label{eq:gene_obs_score1}\\
        \partial_q \ell(\theta; Y) &= - \frac{2 y_1}{r}  - \frac{2p y_2}{p^2 + 2pr}  + \frac{2r y_3}{q^2 + 2qr}  + \frac{y_4}{q} \label{eq:gene_obs_score2}
    \end{align}
    Solving the score equation, $S(\theta) = 0$, thus reduces to solving a system of two polynomials in $p$ and $q$. Since $p$ and $q$ are proportions, we reject any roots outside the unit simplex.

    Differentiating $\ell$ again and multiplying by $-1$ gives the observed data information matrix, $I$. To simplify notation, let $p_y = p^2 + 2pr$ and $q_y = q^2 + 2qr$.
    \begin{align}
        I(\theta;y) &= - \begin{bmatrix}
            \partial^2_p \ell(\theta; y) & \partial_{p,q} \ell(\theta; y)\\
            \partial_{p,q} \ell(\theta; y) & \partial^2_q \ell(\theta; y)
        \end{bmatrix} \mathrm{, where}\\
        \partial^2_p \ell(\theta; y) &=  \frac{2y_1}{r^2} + \frac{2 y_2 (p_y + 2r^2)}{p_y^2} + \frac{4 y_3 q^2}{q_y^2} + \frac{y_4}{p^2}\\
        \partial_{p,q} \ell(\theta; y) &=  \frac{2y_1}{r^2} + \frac{2 y_2 p^2}{p_y^2} + \frac{2 y_3 q^2}{q_y^2}\\
        \partial^2_q \ell(\theta; y) &=  \frac{y_1}{r^2} + \frac{4 y_2 p^2}{p_y^2} + \frac{2 y_3 (q_y + 2r)}{q_y^2} + \frac{y_4}{q^2}
    \end{align}
    The asymptotic standard error of our MLE is $I^{-1}$, evaluated at the estimate.

    \subsection{Complete Data Likelihood, Score and Information}
    \label{app:blood_complete}

    The complete data distribution for our model can be written as follows. Write $\rho_i$ for the probability of genotype $i$. See Table \ref{tab2:blood_type_complete} for the values of these probabilities.
    \begin{align}
        \ell_c(\theta; y,x) &= \log \begin{pmatrix} n \\ x \end{pmatrix} + \sum x_i \log \rho_i(\theta)\\
        & \equiv \sum y_i \log \rho_i\\
        &\equiv 2 x_1 \log r + x_2 \log pr + 2 x_3 \log p + x_4 \log qr + 2 x_5 \log q + x_6 \log pq\\
        &= (2 x_1 + x_2 + x_4) \log r + (x_2 + 2 x_3 + x_6) \log p + (x_4 + 2 x_5 + x_6) \log q\\
        &= n_O \log r + n_A \log p + n_B \log q
    \end{align}
    where $n_O$, $n_A$ and $n_B$ are the number of times allele O, A and B arise respectively in the sampled genotypes. Note that $\ell_c$ depends on $y$ only through $x$, so we suppress $y$ from our notation for complete data quantities. The complete data score function is
    \begin{align}
        S_c(\theta; x) &= \begin{pmatrix}
            \partial_p \ell_c(\theta; x)\\
            \partial_q \ell_c(\theta; x) 
        \end{pmatrix} \mathrm{, where}\\
        \partial_p \ell_c(\theta; x) &= \frac{x_2 + 2 x_3 + x_6}{p} - \frac{2x_1 + x_2 + x_4}{r} = \frac{n_A}{p} - \frac{n_O}{r} \label{eq:comp_score1}\\
        \partial_p \ell_c(\theta; x) &= \frac{x_4 + 2 x_5 + x_6}{q} - \frac{2x_1 + x_2 + x_4}{r} = \frac{n_B}{q} - \frac{n_O}{r} \label{eq:comp_score2}
    \end{align}
    Notice that the score is linear in $x$. To make this relationship explicit, we write $S_c(\theta; x) = \mathscr{S}(\theta) x$, where $\mathscr{S}(\theta) \in \bR^{2 \times 6}$ is a matrix consisting of the coefficients on $x$ in (\ref{eq:comp_score1}) and (\ref{eq:comp_score2}). We will make use of this linearity in Section \ref{app:ASE}.

    Next, we give the information matrix for the complete data.
    \begin{align}
        I_c(\theta;x) &= - \begin{bmatrix}
            \partial^2_p \ell_c(\theta; x) & \partial_{p,q} \ell_c(\theta; x)\\
            \partial_{p,q} \ell_c(\theta; x) & \partial^2_q \ell_c(\theta; x)
        \end{bmatrix} \mathrm{, where}\\
        \partial^2_p \ell_c(\theta; x) &=  \frac{x_2 + 2 x_3 + x_6}{p^2} + \frac{2x_1 + x_2 + x_4}{r^2} = \frac{n_A}{p^2} + \frac{n_O}{r^2}\\
        \partial_{p,q} \ell_c(\theta; x) &=   \frac{2x_1 + x_2 + x_4}{r^2} = \frac{n_O}{r^2}\\
        \partial^2_q \ell_c(\theta; x) &=  \frac{x_4 + 2 x_5 + x_6}{q^2} + \frac{2x_1 + x_2 + x_4}{r^2} = \frac{n_B}{q^2} + \frac{n_O}{r^2}
    \end{align}

    \subsection{Missing Data Distribution}
    \label{app:blood_miss}

    Many quantities which arise in the EM and MCEM algorithms depend on the missing data distribution (i.e.\ the conditional distribution of $X$ given $Y=y$). This distribution is best described componentwise in $X$. First, note that $X_1 = y_1$ and $X_6 = y_4$. Next, we have that $X_2 + X_3 = y_2$ and $X_4 + X_5 = y_3$. Thus, we can write $X_2 |Y=y \sim \mathrm{Bin}(y_2, 2pr / (p^2 + 2pr))$ and $X_4 |Y=y \sim \mathrm{Bin}(y_3, 2qr / (q^2 + 2qr))$. Finally, we recover $X_3$ and $X_5$ by subtracting $X_2$ from $y_2$ and $X_4$ from $y_3$ respectively.

    We make frequent use of the first few conditional moments of $X$, so they are listed here for convenience. Let $\alpha_1 = 2pr / (p^2 + 2pr)$ be the probability parameter for the binomial distribution of $X_2$ given $Y$, and $\alpha_2 = 1 - \alpha_1$. Similarly, let $\beta_1 = 2qr / (q^2 + 2qr)$ correspond to $X_4$ and $\beta_2 = 1 - \beta_1$.
    \begin{align}
        \bE(X | Y=y) &= \begin{pmatrix}
            y_1,  y_2 \alpha_1,  y_2 \alpha_2,  y_3 \beta_1,  y_3 \beta_2,  y_4
        \end{pmatrix}^T\\
        &=: \mu_m\\
        \bV(X | Y=y) &= \begin{pmatrix}
            0 & 0 & 0 & 0 & 0 & 0\\
            0 & y_2 \alpha_1 \alpha_2 & - y_2 \alpha_1 \alpha_2 & 0 & 0 & 0\\
            0 & - y_2 \alpha_1 \alpha_2 & y_2 \alpha_1 \alpha_2 & 0 & 0 & 0\\
            0 & 0 & 0 & y_3 \beta_1 \beta_2 & - y_3 \beta_1 \beta_2 & 0\\
            0 & 0 & 0 & -y_3 \beta_1 \beta_2 & y_3 \beta_1 \beta_2 & 0\\
            0 & 0 & 0 & 0 & 0 & 0
        \end{pmatrix}\\
        &=: \Sigma_m\\
        \bE(XX^T | Y=y) &= \Sigma_m + \mu_m \mu_m^T
    \end{align}
    Conditional expectations of the number of alleles of each kind will be of particular interest.
    \begin{align}
        \nu_O & := \bE(n_O|y)\\
         &= 2y_1 + \frac{y_2 pr}{p^2 + 2pr} + \frac{y_3 qr}{q^2 + 2qr}\\
        &= 2y_1 + y_2 \left( \frac{\rho_2}{\rho_2 + \rho_3} \right) + y_3 \left( \frac{\rho_4}{\rho_4 + \rho_5} \right) &\left( = 2y_1 +  y_2 \left( \frac{\rho_2}{\pi_2} \right) + y_3 \left( \frac{\rho_4}{\pi_3} \right) \right)\\
        \nonumber \\
        \nu_A & := \bE(n_A|y)\\
        &= \frac{2 y_2 pr}{p^2 + 2pr} + \frac{2y_2 p^2}{p^2 + 2pr} + y_4\\
        &= y_2 \left( \frac{\rho_2}{\rho_2 + \rho_3} + \frac{2\rho_3}{\rho_2 + \rho_3} \right) + y_4 &\left(= y_2 \left( \frac{\rho_2}{\pi_2} + \frac{2\rho_3}{\pi_2} \right) + y_4 \right)\\
        &= y_2 \left( 1 + \frac{p^2}{p^2 + 2pr} \right) + y_4\\
        \nonumber \\
        \nu_B &:= \bE(n_B|y)\\
        &= \frac{2 y_3 qr}{q^2 + 2qr} + \frac{2y_3 q^2}{q^2 + 2qr} + y_4\\
        &= y_3 \left( \frac{\rho_4}{\rho_4 + \rho_5} + \frac{2\rho_5}{\rho_4 + \rho_5} \right) + y_4 &\left(= y_3 \left( \frac{\rho_4}{\pi_3} + \frac{2\rho_5}{\pi_3} \right) + y_4 \right)\\
        &= y_3 \left( 1 + \frac{q^2}{q^2 + 2qr} \right) + y_4
    \end{align}

    \subsection{EM Algorithm}
    \label{app:EM}

    In order to apply the EM algorithm, we must construct and optimize the EM objective function. That is, we must compute $Q(\theta|\theta_0) = \bE_{\theta_0} \left[ \ell_c(\theta; y, X) | Y=y \right]$. The EM objective function can be written as
    \begin{align}
		Q(\theta | \theta_0) &:= \bE_{\theta_0} [\ell_c (\theta; X) | Y=y]\\
		&\equiv \nu_O^{(0)} \log r + \nu_A^{(0)} \log p + \nu_B^{(0)} \log q
	\end{align}
    where a superscript zero denotes that the quantity is computed by taking an expectation under $\theta_0$. Differentiating $Q$ with respect to $p$ and $q$ and setting the result to zero, we get the following system of equations:
    \begin{align}
        \frac{\nu_A^{(0)}}{p} = \frac{\nu_O^{(0)}}{r} \label{eq:blood_update1}\\
        \frac{\nu_B^{(0)}}{q} = \frac{\nu_O^{(0)}}{r} \label{eq:blood_update2}
    \end{align}
    This system of equations can be used to solve for a fixed point of the EM algorithm by evaluating $\nu_O$, $\nu_A$ and $\nu_B$ at $\theta$ instead of $\theta_0$. Note that the fixed point equations which result from this substitution exactly match the observed data score equations given by equations (\ref{eq:gene_obs_score1}) and (\ref{eq:gene_obs_score2}). Indeed, this relationship holds in general under mild conditions \citep{Wu83}.

    \subsection{Asymptotic Standard Error}
    \label{app:ASE}

    Recall that the EM algorithm computes the MLE, which has asymptotic covariance matrix equal to the inverse Fisher information matrix evaluated at the true parameter value. In practice, we estimate this covariance with the inverse of the observed information matrix evaluated at the MLE. Using Proposition \ref{thm2:info_decomp}, we can calculate the observed information matrix using conditional expectations of quantities derived from the complete data likelihood. 
    
    To this end, we need to evaluate the conditional expectations in expression (\ref{eq:info_at_MLE}) of Proposition \ref{thm2:info_decomp}. It is convenient for us to write $S_c(\theta) =: \mathscr{S}(\theta) X$ (see Appendix \ref{app:blood_complete}). Then 
    \begin{align}
        I_c(\hat{\theta}) &= \begin{bmatrix}
            \frac{\nu_A}{p^2} + \frac{\nu_O}{r^2} & \frac{\nu_O}{r^2}\\
            \frac{\nu_O}{r^2} & \frac{\nu_B}{q^2} + \frac{\nu_O}{r^2}
        \end{bmatrix} \mathrm{, and}\\
        \bE_{\hat{\theta}} [ S_c(\hat{\theta}) S_c(\hat{\theta})^T | Y=y] &= \mathscr{S}(\hat{\theta}) \bE_{\hat{\theta}} \left[ X X^T | Y=y \right] \mathscr{S}(\hat{\theta}) \\
        &= \mathscr{S}(\hat{\theta}) (\Sigma_m + \mu_M \mu_M^T) \mathscr{S}(\hat{\theta})\\
    \end{align}
    While it is possible to expand the above expressions, they quickly become too long to easily interpret. We instead leave these as computational formulas and use them as a guide for writing \texttt{R} or \texttt{Julia} code.

\end{appendices}

\newpage

\bibliographystyle{plainnat}
\bibliography{mybib}

\end{document}